\documentclass[letterpaper,twocolumn,10pt]{article}
\usepackage{usenix-2020-09}

\usepackage[pdftex]{graphicx}
\usepackage[utf8]{inputenc}
\usepackage{colortbl}
\usepackage{xargs}  
\usepackage{afterpage}
\usepackage[tight,normalsize,sf,SF]{subfigure}
\usepackage{booktabs}
\usepackage{array}
\usepackage{paralist}
\usepackage{verbatim}
\usepackage{arydshln}
\usepackage{color}
\usepackage{float}
\usepackage{url}

\usepackage{tikz}
\usepackage{amsmath}
\usepackage{amsfonts}
\usepackage{amssymb}
\usepackage{amsmath}
\usepackage{amsthm}

\usepackage{lscape}
\usepackage{enumitem}
\usepackage{msc}
\usepackage{algorithm}
\usepackage[noend]{algpseudocode}

\newtheorem{definition}{Definition}
\newtheorem{theorem}{Theorem}

\newcommand{\commentOut}[1]{}
\algnewcommand{\LineComment}[1]{\State \(\triangleright\) #1}

\newcommand{\CV}{\mathcal{V}_{corr}}
\newcommand{\V}{\mathcal{V}}
\newcommand{\HV}{\Tilde{\mathcal{V}}}

\usepackage[disable]{todonotes} 
\newcommandx{\VTNote}[2][1=]{\todo[linecolor=green,backgroundcolor=green!25,bordercolor=green,#1]{VJT:#2}}
\newcommandx{\THNote}[2][1=]{\todo[linecolor=orange,backgroundcolor=orange!25,bordercolor=orange,#1]{TEH:#2}}
\newcommandx{\EMNote}[2][1=]{\todo[linecolor=blue,backgroundcolor=blue!25,bordercolor=blue,#1]{EEM:#2}}

\usepackage{bm}
\restylefloat{table}
\definecolor{green}{RGB}{144,238,144}
\definecolor{blue}{RGB}{135,206,250}
\definecolor{yellow}{RGB}{255,255,0}
\definecolor{orange}{RGB}{255,140,0}
\definecolor{white}{RGB}{255,255,255}
\definecolor{red}{RGB}{220,0,0}

\newcommand{\Vote}{\mathit{Vote}}
\newcommand{\Mac}{\mathit{MAC}}
\newcommand{\Wbb}{\mathit{WBB}}
\newcommand{\Param}{\mathit{Params}}
\newcommand{\Paper}{\mathit{Paper}}
\newcommand{\DecProof}{\mathsf{PrfDec}}
\newcommand{\ReceivedVote}{\mathit{ReceivedVote}}
\newcommand{\commentOutForESORICS}[1]{}
\newcommand{\commit}{\mathsf{Com}}
\newcommand{\PrfEnc}{\mathsf{PrfEnc}}
\newcommand{\PrfKnow}{\mathsf{PrfKnow}}
\newcommand{\rerand}{\mathsf{Rerand}}

\newcommand{\Result}{\mathsf{Result}}

\newcommand{\negl}{\mathsf{negl}}
\newcommand{\dlog}{\mathsf{dlog}}
\newcommand{\Mix}{\mathsf{Mix}}
\newcommand{\Dec}{\mathsf{Decrypt}}
\newcommand{\PEP}{\mathsf{PlaintextEquivalent}}
\newcommand{\vid}{\textit{VoterID}}
\newcommand{\receivedvid}{\textit{RecVoterID}}
\newcommand{\shorten}[1]{}

\usepackage[utf8]{inputenc} 

\newcommand{\ifanon}[2]{#2}

\newcommand{\ifcompressing}[2]{#1}

\usepackage{booktabs} 
\usepackage{array} 
\usepackage{paralist} 
\usepackage{verbatim} 
\usepackage{arydshln}

\begin{document}
\ifanon{
	\title{Towards Verifiable Remote Voting with Paper Assurance}
	\author{Anonymous}
}	
{	
	\title{Verifiable Remote Voting with Paper Assurance\thanks{\url{elem0@protonmail.com}, \url{xavier.boyen@qut.edu.au}, \url{chris@culnane.org}, \url{kristian.gjosteen@ntnu.no}, \url{thomas.haines@anu.edu.au}, \url{vanessa.teague@anu.edu.au}}}
\author{  
{\rm  Eleanor McMurtry}\\
Department of Computer Science\\
ETH Zurich
\and
{\rm Xavier Boyen}\\
QUT
\and
{\rm Chris Culnane}\\
Castellate 
\and
{\rm Kristian Gjøsteen}\\
NTNU
\and
{\rm Thomas Haines}\\
ANU
\and
{\rm Vanessa Teague}\\
Thinking Cybersecurity Pty. Ltd.\\ and ANU }

}


%
%
%

\maketitle

\begin{abstract}
We propose a protocol for verifiable remote voting with paper assurance. It is intended to augment existing postal voting procedures, allowing a ballot to be electronically constructed, printed on paper, then returned in the post. It allows each voter to verify that their vote has been correctly cast, recorded and tallied by the Electoral Commission. The system is not end-to-end verifiable, but does allow voters to detect manipulation by an adversary who controls either the voting device, or (the postal service and electoral commission) but not both.
The protocol is not receipt-free, but if the client honestly follows the protocol (including possibly remembering everything), they cannot subsequently prove how they voted.
Our proposal is the first to combine plain paper assurance with cryptographic verification in a (passively) receipt-free manner.
\end{abstract}
\vspace{-0.3cm}
\section{Introduction}
\vspace{-0.2cm}
The biggest form of remote voting is postal voting---nearly half the US presidential votes in 2020 were cast by mail~\cite{pew:votingExperience}. Postal voting suffers from the same coercion problems that all remote voting approaches do, but it provides easy cast-as-intended verification via the intuitive checking of the plaintext ballot. This simplicity greatly favours postal voting over online remote voting. Most online voting schemes require a level of trust in the device for integrity or involve complex and difficult cast-as-intended verification procedures. Failure to conduct the audits can have a catastrophic impact on integrity.  However, postal voting produces no evidence that the vote was accurately included and tallied---this is what we add.

In an ideal world remote voting in any form would be deployed only where absolutely necessary, for example to enfranchise a house-bound voter. Unfortunately, there is a global trend towards increases in remote voting, either by mail \cite{rallings20102010,VEC_PostalVoting_Position,commonsbriefing} or worse, online voting \cite{gjosteen2011norwegian, stenerud2012reality, heiberg2014verifiable}. 
The global coronavirus pandemic has accelerated the trend.

Postal voting has seen demonstrated instances of fraud in recent years \cite{UK_Court_Judgement, telegraph, TheAgeMoreland}; those that have been detected may be only a small indication of a larger problem.  There have been repeated problems with online voting systems\cite{halderman2015new, springall2014security,culnane2017trust, specter2020ballot, specter2021security, gaudry2020breaking}, including flaws in cryptographic verification mechanisms that have been demonstrated only after they had been relied upon in an election~\cite{haines2020not}. Resilience against cryptographic failure is another great advantage of a plain paper backup---the continued push towards paperless remote voting runs counter to both academic opinion and the established experiences of deployed systems. 
However, Electoral Commissions (ECs) are under pressure to deliver some form of remote voting. 


We propose a new remote voting system that combines electronic construction of the ballot with a familiar and easy-to-verify paper record. Using their own device, the voter generates a ballot, prints it, and returns it by post. This halves the use of the postal channel, thus increasing the time available for voters to construct and cast their ballots. The approach offers immediate plaintext cast-as-intended verification to the voter, with an option for verifying the vote was recorded and tallied properly if the voter conducts a simple electronic check that her vote is properly included on the bulletin board.

The system is not end-to-end verifiable, but it provides verifiability against an adversary who controls \textbf{either} the voter's device \textbf{or} (the postal system and the EC) but not both.  The verification protocols detect manipulation by an adversary who corrupts the EC and the post, as long as the client's randomness remains secret. Although it is not Receipt-Free, clients who execute the protocol honestly (including remembering their randomness) cannot prove how they voted.

Our proposal offers a new
choice of tradeoffs: easy cast-as-intended verification, defence against a cheating post and EC (unlike traditional postal voting) and some defence against coercion (though it does introduce coercion opportunities not present in traditional postal voting).
Whilst our recommendation is not to conduct remote voting unless absolutely necessary, this may be the best choice in some scenarios.



\subsection{Our contribution / protocol properties}
The scheme resembles existing end-to-end verifiable voting schemes, except that ``the system'' whose behaviour the voter needs to verify includes her device, the postal service, and the EC's vote-receiving process.  We use a web bulletin board (WBB) (such as \cite{culnane2014peered}), which is an authenticated broadcast channel with memory.  The voter verifies that her vote has been cast as she intended by reading it on a plain paper printout, which she puts in a post box along with other verification artifacts.  At the end of the election, the voter checks that her ID appears among the confirmed votes on the WBB.  
The proof of proper tallying is universally verifiable.

A great advantage of our scheme is that it falls back to postal voting integrity guarantees even if all the electronic devices are compromised.  If voters check their printouts, and the postal service can be trusted, and the processes for opening the envelopes and counting the ballots are properly observed, then the election outcome is correct.  We refer in the text to the places where scrutineers may watch the paper processing to gain the evidence they need to have trust in the paper-only tally, independent of any cryptography. The cryptographic protocol adds the option for these processes to be verified electronically by people who are not present at the counting location or do not trust the post.  Scrutineers are welcome, but not required, to do anything to support the security claims of the cryptographic protocol. 

Of course, the scheme also adds the possibility to fabricate problems, for example by submitting inconsistent values to make it appear that there was cheating when there was not---a correct result may look suspicious.  Accountability (and other defences against this) is a topic for future work.

\emph{Receipt Freeness}~\cite{benaloh1994receipt} means that the system does not allow voters to prove how they voted.  Our system offers a weaker version we call \emph{honest-but-remembering Receipt Freeness}---if a voter's device executes the protocol honestly, she cannot subsequently prove how she voted (assuming her mail isn't read and her channel to the EC is not tapped), even if the device remembers the randomness used to generate the ciphertexts. Thus it is strictly better than Helios (which does not claim to be Receipt-Free). She can, however, produce a receipt by actively deviating from the protocol.  For example, a voter who posts commitments given to her by the coercer can later prove to that coercer how she voted. In this sense our receipt freeness property is weaker than that proposed by Benaloh~\cite{benaloh1994receipt} and proven for some attendance polling-place systems~\cite{moran2006receipt}---see Section~\ref{subsec:RF}. Note this is also true for traditional postal voting e.g. if the voter films her voting process.
We do not claim this is sufficient for government elections. It is, however, better than the coercion-resistance properties of any remote end-to-end verifiable e-voting system.

Our proposal is the first remote voting system to combine plain paper assurance with cryptographic verification in an honest-but-remembering receipt-free manner.

Our protocol has the following security properties:

\begin{itemize}
	\item privacy from an adversary that does not collude with the voting client, post, or EC, given threshold trust for vote decryption and proper opening of paper ballots,\VTNote{Actually this is an aspriational privacy property - this is what I'd like, but at the moment the EC is trusted for privacy.}
	\item honest-but-remembering receipt-freeness against an adversary who sees the WBB but does not tap the voter-EC communication channels or collude with the EC, 
	\item easy cast-as-intended verifiability based on plain paper,
	\item recorded-as-intended verifiability secure against an attacker who controls (the post and the EC) or the voter's device, but not both.

\end{itemize}

This is the first proposal to include all four of these advantages. Correct tallying is universally verifiable. If verification is properly performed, the proof of integrity is significantly better than postal voting.

A complete prototype implementation, including voting, tallying, and verifying, is available at:\\
\ifanon{\url{https://www.dropbox.com/s/diige6wd69772a9/}}{\url{https://github.com/eleanor-em/papervote}}

\subsection{Protocol main idea}
The key innovation is in the recorded-as-cast step, which allows the voter to check that the vote she put in the post was correctly recorded on the WBB.  We use a Carter and Wegman \cite{carter1979universal} universal hash on the WBB to bind the vote without allowing people to prove how they voted.  Before voting, the voter's device randomly generates two secrets $a$ and $b$ in $\{1,\ldots,q-1\}$, where $q$ is a large prime known to all voters.  The device posts a Pedersen commitment~\cite{pedersen1991non} to these secrets on the WBB.
When a voter wishes to cast a ballot, her device computes a MAC as $\Mac=a\cdot\Vote+b \bmod q$ and sends the MAC and vote, in encrypted form, to the EC.  The EC re-randomises the encryptions (for receipt-freeness) and posts them on the WBB. 

The voter sends two pieces of paper to the EC by post.
\begin{description}
	\item[Paper 1] contains her plaintext vote and encrypted commitment openings (with a proof of knowledge). 
	\item[Paper 2] contains her plaintext $\vid$ in human- and machine-readable form.
\end{description} 
These do not have to be made by one device, and splitting the task could improve privacy---this is explored in Section~\ref{sec:extensions}. 

As well as printing her $\vid$ and sending it with her paper vote, the voter does whatever is usual for postal voting in her country, such as writing her name and address on an outer envelope or signing it.  This is used at the EC for checking eligibility and identity against the electoral roll.

The Carter-Wegman hash provides two useful properties:
\begin{itemize}
	\item Without knowing $a$ and $b$, an attacker cannot generate a valid alternative $(\Vote, \Mac)$ pair except with very small probability. 
	\item Even after $a$ and $b$ are exposed, a voter can plausibly claim to have cast any vote (if the system does not reveal $\Vote$ or $\Mac$ to the coercer)---she simply claims the correct $\Mac$ for whichever $\Vote$ the coercer demands.
\end{itemize}

The main idea of our verifiability proof is that anyone who intercepts the envelope, including a corrupt EC, cannot (except with small probability) change the MAC and the plaintext vote consistently, unless they can guess the values of $a$ and $b$ before the EC posts the vote and MAC on the WBB.  Using perfectly hiding commitments for $a$ and $b$ means that they are hidden unless the client exposes them---obviously this assumes the client keeps them secret. This protects the vote from manipulation by a corrupt postal service or EC, even if all the decryption authorities collude. So EC-re-randomization provides honest-but-remembering receipt freeness, while the Carter-Wegman hash prevents dishonest EC re-randomization. The homomorphic property of the encryption scheme is then used to recreate the right $\Mac$ from the paper vote received in the mail---if they match, the vote is accepted.

Thus verification depends on a secrecy assumption, but one in which the voter generates their own secret.  The protocol is not end-to-end verifiable, but it contrasts favourably with protocols such as code voting (described below) in which integrity depends on the secrecy of values that are generated centrally and sent to the voter.


\subsection{Related work on remote recorded-as-intended verification and receipt freeness} \label{subsec:relatedWork}
Neither the scientific literature nor the remote electronic systems used in practice have good solutions for cast-as-intended verification. They are either too hard for ordinary voters to use easily, or they are dependent on a secrecy assumption that is unverifiable and outside the voter's control. Ordinary postal voting has clear and simple cast-as-intended verification (assuming we take a vote to be ``cast'' when it is put in the mailbox) but no recorded-as-cast verification at all.

\VTNote{Add Demos-D and Mark Ryan's extension of JCJ for cast-as-intended verification?}

The Helios voting system~\cite{Adida08helios:web-based} offers end-to-end verifiability in a remote all-electronic setting.  A diligent voter gets very good evidence that her vote is cast as she intended and properly included, followed by a universally verifiable count.  \ifcompressing{}{Helios has been used successfully in university elections~\cite{adida2009electing} and by the International Association for Cryptologic Research. }  However, the verification is difficult enough that ordinary voters may be tricked into not performing it properly~\cite{karayumak2011usability}, and the recommended challenge strategies do not form Nash equilibria in a remote setting~\cite{culnane2016strategies}.  Even more importantly, Helios is not (and has never claimed to be) \emph{receipt free}: a voter can prove how she voted if her client remembers the randomness used to encrypt her vote. Thus Helios asumes low-coercion environments, and diligent voters who perform the cast-as-intended verification step on an independent device. The Estonian Internet voting system~\cite{springall2014security, heiberg2020planning} is similar.

The Civitas Internet voting system~\cite{clarkson2008civitas}, based on a protocol by Juels {\it et al.} \cite{juels2005coercion}, provides a very strong form of coercion resistance but no cast-as-intended verification. \ifcompressing{}{It is suitable for simple plurality and Condorcet-style voting systems and can be cleverly extended to Borda counting, but does not work for more complex systems such as Instant Runoff Voting (IRV) and Single Transferable Vote (STV).}
The Selene Internet voting system~\cite{ryan2016selene} can be used to enhance it with cast-as-intended verification~\cite{iovino2017using}, or as an adjunct to Helios-style systems.  In Selene, election trustees generate a unique tracker for each user, who uses it later to verify that their  (plaintext, electronic) vote was correctly included. 

In \emph{Code Voting} systems, the voter verifies cast-as-intended and recorded-as-cast in a single step, using a code sheet sent (usually) by paper mail. Remotegrity~\cite{zagorski2013remotegrity} is a remote version of the Scantegrity II voting system, with extra codes for confirming a properly-verified vote. In Pretty Good Democracy~\cite{ryan2009pretty} (PGD), there are codes for sending the vote and only one return code, to acknowledge receipt. 

In \emph{Code-return} voting systems, such as the Norwegian~\cite{gjosteen2011norwegian} and Swiss Internet voting systems, voters cast an encrypted vote and then receive a confirmation code, which they check against a code sheet received in the mail.   All these systems allow some subsequent verification of the tally.

Although code voting and code-return voting are convenient, they suffer from a major drawback: the \emph{integrity} of the outcome is dependent on the \emph{secrecy} of the codes.  In principle, secrecy is impossible to verify.  Printing the code sheets securely and sending them privately is the major practical problem in these schemes. Thus these systems are not really end-to-end verifiable.  (Ours is also not end-to-end verifiable, because verification is also dependent on a secrecy assumption, but one in which the voter generates their own secret.) 

Although Remotegrity and PGD have been adapted for instant runoff voting (also called ranked-choice voting), code voting becomes unwieldy as the number of preferences grows.

A plain paper mail step could also be added to the Internet voting solutions described above.  With Helios, this would result in properties incomparable to our scheme: genuine end-to-end verifiability, a simple cast-as-intended step for the paper backup, but no receipt freeness.  

It seems less useful to combine code-voting-style solutions with a plain paper return.  Firstly, it requires paper mail in both directions, thus removing much of the benefit of an electronic solution.  Second, it doesn't solve the problem that a malicious authority (or sufficient collusion among trustees) can fabricate a successful-looking verification.  This seems strictly worse than our solution, in which even a completely corrupted EC cannot cheat undetectably unless it also controls the client.

Pr\^et \`a voter~\cite{ryan2009pret} uses preprinted auditable ciphertexts which the voter selects or arranges to express their vote.  Although designed for pollsite voting, the idea could be extended to a remote setting, but auditing the printouts would be cumbersome.
Belenios VS~\cite{cortier2019beleniosvs} extends on this idea in a remote setting, allowing voters to receive their preprinted ciphertexts by mail and verify them with a device that is assumed to be independent of their voting client.  It also incorporates eligibility verifiability and receipt freeness. Its main disadvantage is that all its security properties depend on non-collusion between the registration server and the voting server\ifcompressing{}{, whereas our protocol is verifiable even if the server-side is completely corrupted}.  However, Belenios VS preserves privacy even against a corrupted voting client. (Our system can be expanded to do so\ifcompressing{}{at the cost of some extra complexity}---see Section~\ref{sec:extensions}.)

Several existing designs combine plain paper ballots with cryptographic verification for pollsite voting~\cite{rosen2011wombat,bell2013star,chaum2008scantegrity,benaloh:electionguard}.  
These are not designed for a remote setting---the aim of this paper is to take that design philosophy to remote voting.

Also note that the usability of even the simplest cast-as-intended verification mechanisms is questionable, with practical failure rates shown for both code voting~\cite{kulyk2020towards} and the checking of plaintext printouts~\cite{bernhard2020can, demillo2018voters, kortum2020voter,everett2007usability}.

\shorten{
Other internet voting systems with interesting combinations of properties are {\it caveat coercitor}~\cite{grewal2013caveat}, which allows voters 
to use a repeat vote as a signal that they have been coerced, and Du-vote~\cite{grewal2015vote, kremer2016or} which depends on a partially-trusted hardware module.}

In summary, no existing solution provides both receipt-freeness and easily-usable recorded-as-intended verification in a remote setting,
while protecting integrity against a fully corrupt authority.
Our contribution is to fill this gap. 
Compared with code-return systems, our scheme has two important differences: the code secrecy assumption is on the client, not the electoral authorities, and the system easily accommodates arbitrary ballots. 
Table~\ref{tab:comparison} compares the properties of our proposal to other schemes used or proposed.

Verifiable postal voting~\cite{benaloh2013verifiable} is closest to our setting.  This proposal improves on that work by achieving
 honest-but-remembering receipt freeness, and having
a much higher probability of detecting manipulation.

Complex voting schemes raise extra challenges for election privacy and verification.  Aditya {\it et al.} \cite{aditya2003secure, aditya2004efficient} first examined cryptographic election verification for instant runoff elections.  We reuse their idea of getting an authority to re-randomize votes before publication in order to achieve Receipt Freeness, but in our scheme the authority does not need to be trusted to re-randomize honestly, as long as the client keeps its secrets.

\begin{figure*}
	\begin{Tabular}{lllllllll}
		{\bf Protocol} & {\bf Cast-as-intended} & {\bf Recorded-as-cast} & {\bf Verifiable vs} & {\bf Receipt } & {\bf Paper}  & {\bf Complex}\\
		& {\bf verification} & {\bf verification} & {\bf fully corrupt EC} & {\bf Freeness} & {\bf directions} &  {\bf ballots} \\
		Postal voting & \color{blue} hand-marking & \color{red} No & \color{red} No & \color{blue} Yes & \color{red}2 & \color{blue} Yes \\
		Code Voting & \color{orange} EC code secrecy & \color{orange} EC code secrecy & \color{red} No & sometimes & \color{green} 1 & \color{red} No   \\
		Code-return & \color{orange} EC code secrecy & \color{orange} EC code secrecy & \color{red} No & sometimes & \color{green} 1 & \color{red} No   \\
		Code-return + paper & \color{blue} printout reading & \color{orange} EC code secrecy & \color{red} No & sometimes & \color{red} 2  & \color{red} No  \\
		Helios/Estonia & \color{green} Independent device & \color{green} Independent device & \color{green} Yes (Helios) & \color{red} No & \color{blue} 0 & \color{red}No (*)\\
		Helios + paper & \color{blue} printout reading & \color{green} Independent device & \color{green} Yes & \color{red} No & \color{green} 1  & \color{red} No \\
		Belenios-VS & \color{green} Independent device & \color{green} Independent device & \color{red} No & \color{blue} Yes & \color{green} 1 & \color{red} No \\
		
		Our system & \color{blue} Printout reading & \color{orange} Device code secrecy & \color{green} Yes & \color{green} Honest-but- & \color{green} 1 & \color{blue}Yes \\ 
		
		& & & & \color{green} remembering & & \\
	\end{Tabular} \label{tab:comparison}
	\caption{Comparison with existing systems. 
		There are no good solutions, but our system offers a combination of modest trust assumptions, simple verification and honest-but-remembering receipt freeness that is not otherwise available. (*) Note Helios and the Estonian system could be easily adapted to complex ballots if they adopted a mixnet instead of homomorphic tallying.}
\end{figure*}

\subsection{Structure of this Paper}
Cryptographic tools are described in the next section, followed by the protocol  in Section~\ref{sec:protocol}, including algorithms for casting, receiving, and tallying votes. Sketched security arguments are given in Section~\ref{sec:securityProofs}, with formal proofs in the Appendix.  
A prototype implementation and small trial are discussed in Section~\ref{sec:implementation}.
In Section~\ref{sec:extensions} we discuss some simple extensions, while limitations and further work are described in Section~\ref{sec:limitations}.

\section{Cryptographic Background} \label{sec:crypto}

Let $T$ be a set of election trustees. We use the following:
\paragraph{ElGamal encryption scheme}
The encryption scheme has parameters $\mathbb{G} = (G, g, q)$ where $G$ is a cyclic group of prime order $q$ generated by $g$.\footnote{The particular group chosen does not matter as long as the decisional Diffie-Hellman (DDH) problem is hard in $G$.} A secret key $sk$ is jointly generated amongst the trustees $T$ using $k$-out-of-$n$ Pedersen key generation~\cite{pedersen1991threshold} as an extension of Shamir secret sharing~\cite{shamir1979share}, with corresponding public key $pk$.
The scheme encrypts a message $m$ by setting $e = (g^r, m\cdot(pk)^r)$ for a blinding factor $r$. We will sometimes use $m = g^v$ to encrypt some value $v$---this allows for homomorphic addition of encrypted data by elementwise multiplication of ciphertexts.

We write $\{m\}_{pk}$ for an ElGamal encryption of message $m$ with public key, optionally producing (as defined in~\cite{bernhard2012not}):
	\begin{description}
		\item[$\PrfEnc_{\mathbb{G}, pk}(m, c)$:] an adaptively secure noninteractive zero knowledge proof (ZKP) that $c$ is an encryption of $m$;
		\item[$\PrfKnow_{\mathbb{G}, pk}(e=\{m\}_{pk})$:] an adaptively secure noninteractive ZKP of knowledge of the message $m$.
	\end{description}
We write $\{m_1, m_2\}_{pk}$ to mean multiple encryptions of different plaintexts $m_1$ and $m_2$.
We also sometimes abuse notation and write $\PrfKnow(\mathcal{S})$ for a vector of ciphertexts $\mathcal{S}$ to mean a vector of proofs of knowledge, one for each element in $\mathcal{S}$.

Threshold decryption of a ciphertext $\{m\}_{pk}$ (as defined in~\cite{cramer1997secure})  produces the plaintext $m$ and a universally-verifiable proof of proper decryption $\DecProof_{\mathbb{G}, pk}(e,m)$, which is a vector of at least $t$ adaptively secure NIZKs of equality of discrete logarithms.
We write $\Dec_{\mathbb{G}, pk}(\{m\})$ to mean a pair of a plaintext and its decryption proof.

ElGamal ciphertexts may be \emph{re-randomised}, written $\rerand_{\mathbb{G}, pk}((e_1,e_2))$ which means generating a random $r$ s.t. $1\leq r\leq q$ and setting $\rerand_{\mathbb{G}, pk}((e_1,e_2)) = (e_1 g^{r},e_2 pk^{r})$. This produces a new encryption of the same ciphertext.

Finally, we use \emph{plaintext equivalence proofs} (PEPs), a universally-verifiable ZKP that two ciphertexts encrypt the same message, as defined in~\cite{cryptoeprint:2020:909} based on~\cite{jakobsson2000mix}.
\paragraph{Pedersen commitment}
The commitment scheme~\cite{pedersen1991non} has parameters $\mathbb{P} = (G, h_1, h_2)$ where $G$ is a cyclic group in which discrete logarithms are hard, and $h_1, h_2$ are generators chosen such that nobody knows $\dlog_{h_1} h_2$ (as described in~\cite{kerry2013fips}). We write $\commit_{\mathbb{P}}(a;r_a)$ to mean a Pedersen commitment to the value $a$ using randomness $r_a$ (i.e. the value $h_1^a\cdot h_2^{r_a}$).

\paragraph{Mixing}
We write $\Mix_{\mathbb{G}, pk}(\mathcal{S})$ to mean a universally-verifiable distributed mix of the vector $\mathcal{S}$ (with the associated ZKPs), as in~\cite{verificatum2018,haines2019description}. Each trustee performs one stage of the mix, so that as long as at least one trustee is honest, the resulting mix is private (the link between inputs and outputs is unknown). \textbf{Note:} If $\mathcal{S}$ contains plaintexts, they will be encrypted before the mix process.

\VTNote{This and the para below can probably be removed, or rather subsumed into more precise def's above.}

\paragraph{Web bulletin board}
We model the WBB as a public broadcast channel with memory.
That is, items cannot be removed from the bulletin board once they are published, and every participant's final view of the WBB is identical~\cite{culnane2014peered,hirschi2020fixing} . In practice this means that we need the voter to have access to the WBB via a channel independent  from the Client Device.
We assume that it is available both during and after the election period, implying that a malicious authority cannot block uploads (though a malicious client might fail to upload).

\section{The Protocol} \label{sec:protocol}
Assume a list of $\vid$s constructed so that
\begin{itemize}
	\item each voter can recognise their own $\vid$, and
	\item no two voters have the same $\vid$.
\end{itemize}

The latter assumption is important for preventing clash attacks~\cite{kusters2012clash}, in which two voters are convinced that the same entry on the bulletin board is theirs.
\VTNote{Align these assumptions with later comments on whether (a) everyone should be able to verify that all VoterIDs correspond to eligible voters; (b) nobody should be able to guess a valid voter ID for someone else (to prevent ballot stuffing).  Neither of these is part of our core security claims, but both should be discussed, prob at the end.}

The WBB rows are indexed by $\vid$s---each voter updates their own row.
In practice, a good system would include an authentication mechanism to prevent voters writing to others' rows---we do not detail that here because it is separate from verifiability. $\vid$s could be a simple function of the voter's name and address. If a  voter uses someone else's $\vid$ we assume it can be detected (by the owner of that ID). \ifcompressing{}{See Section~\ref{sec:extensions} for more options.} For eligibility verifiability, we need to assume that the public has some way of assessing whether a VoterID corresponds to an eligible voter.  


Recall the election secret key $sk$ is shared amongst the set of electoral trustees $T$. This is important for voter privacy: we will assume that $n - k + 1$ of the trustees are honest so that no $k$ of them collude to decrypt data they are not supposed to. \ifcompressing{}{This way, a lone cheating trustee cannot identify how anybody voted.}

Setup is shown in Algorithm~\ref{fig:setup}. The electoral commission prepares a set of encrypted $\vid$s, and places them inside a smaller envelope. The smaller envelope is placed in a larger envelope marked with the corresponding plaintext.
The double-envelope is called $D_\vid$. It can be generated in advance or on demand, and will be used during vote receiving to link voters to votes in a private manner. Note that the mechanism that generates these double-envelopes is the \emph{only} part of the EC that is trusted for privacy.

\begin{algorithm} 	\caption{\textit{Setup:} System setup protocol}
	\begin{algorithmic}[1]
\LineComment{The following are posted to the WBB:}
\State $(\vid_i)$: a list of IDs of eligible voters.  We assume that these are assigned one-on-one to each voter.

\State $\mathbb{G} \gets (G, g, q)$: the public parameters of an ElGamal encryption scheme as discussed in Section~\ref{sec:crypto}.
\State $pk$: an ElGamal public key generated jointly among the trustees with the ElGamal parameters $\mathbb{G}$ and security parameter $\lambda$. The corresponding secret key $sk$ is shared among the trustees. This may be done e.g. following~\cite{pedersen1991threshold}.
\State	$\mathbb{P} \gets (P, h_1, h_2)$: the public parameters of a Pedersen commitment scheme  as discussed in Section~\ref{sec:crypto}.
\LineComment{The EC creates the following double envelopes:}
\State Outer envelope of $D_{\vid_i} \gets \vid$
\State Inner envelope of $D_{\vid_i} \gets e_{\vid_i} = \{\vid_i\}_{pk}$
   
\end{algorithmic}
	\label{fig:setup}

\end{algorithm}

\subsection{Voting} \label{subsec:voting}
The voter's experience is extremely straightforward and is detailed in Algorithm~\ref{fig:genAndCasting} (\textit{Cast}). To generate a ballot, the voter's device chooses two random secrets $a, b,\in G$, and randomness $r_a, r_b$, then publishes commitments $c_a = \commit_\mathbb{P}(a,r_a)$ and $c_b = \commit_\mathbb{P}(b,r_b)$ on the WBB (Steps~\ref{step:cast:secrets}--\ref{Step:VoterCommit}). When the voter chooses a vote, the device computes a message authentication code $\Mac = a\cdot\Vote+b$,\footnote{To encode votes as an integer, the EC could e.g. publish an ordered list of choices indexed from 0 to represent candidates or more complex preferences.} and sends encryptions of the vote and MAC to the EC, which re-randomises the values and posts them on the WBB indexed by the $\vid$ (Steps~\ref{step:cast:choosevote}--\ref{Step:ECPostsVoteMAC}). After the EC has posted these values, the device prints two pieces of paper (Steps~\ref{step:cast:wait}--\ref{step:cast:papers}):

\begin{description}
	\item[$\Paper_1$] contains the human-readable plaintext vote, encryptions of $a,b,r_a,r_b$ (i.e. the commitment openings) and proofs of plaintext knowledge of those ciphertexts.
	\item[$\Paper_2$] contains the human-readable plaintext $\vid$. 
\end{description}

Example printouts from our prototype are in Figures~\ref{fig:paper1} and~\ref{fig:paper2}. They use QR codes for the non-human-readable values.

The voter must mail both to the EC (using the standard postal voting procedures in her country, which may include signing the envelope) by placing $\Paper_1$ in a smaller envelope, and both the smaller envelope and $\Paper_2$ in the larger envelope. This mirrors standard postal voting practices.

The voter must check that $\Paper_1$ contains a correct human-readable printout of her vote, and $\Paper_2$ contains a correct human-readable printout of her $\vid$. If she wants to check that her vote has not been dropped, she needs to visit the WBB after the election to check that her $\vid$ is in the list of included IDs.  Note that she does \emph{not} have to do anything to verify that the QR codes on her printouts are not maliciously generated---this will be detected by subsequent verification (assuming the attacker model given in the Introduction).
\
\begin{figure}
	\centering
	\includegraphics[scale=0.4, clip, trim=0 300 0 0]{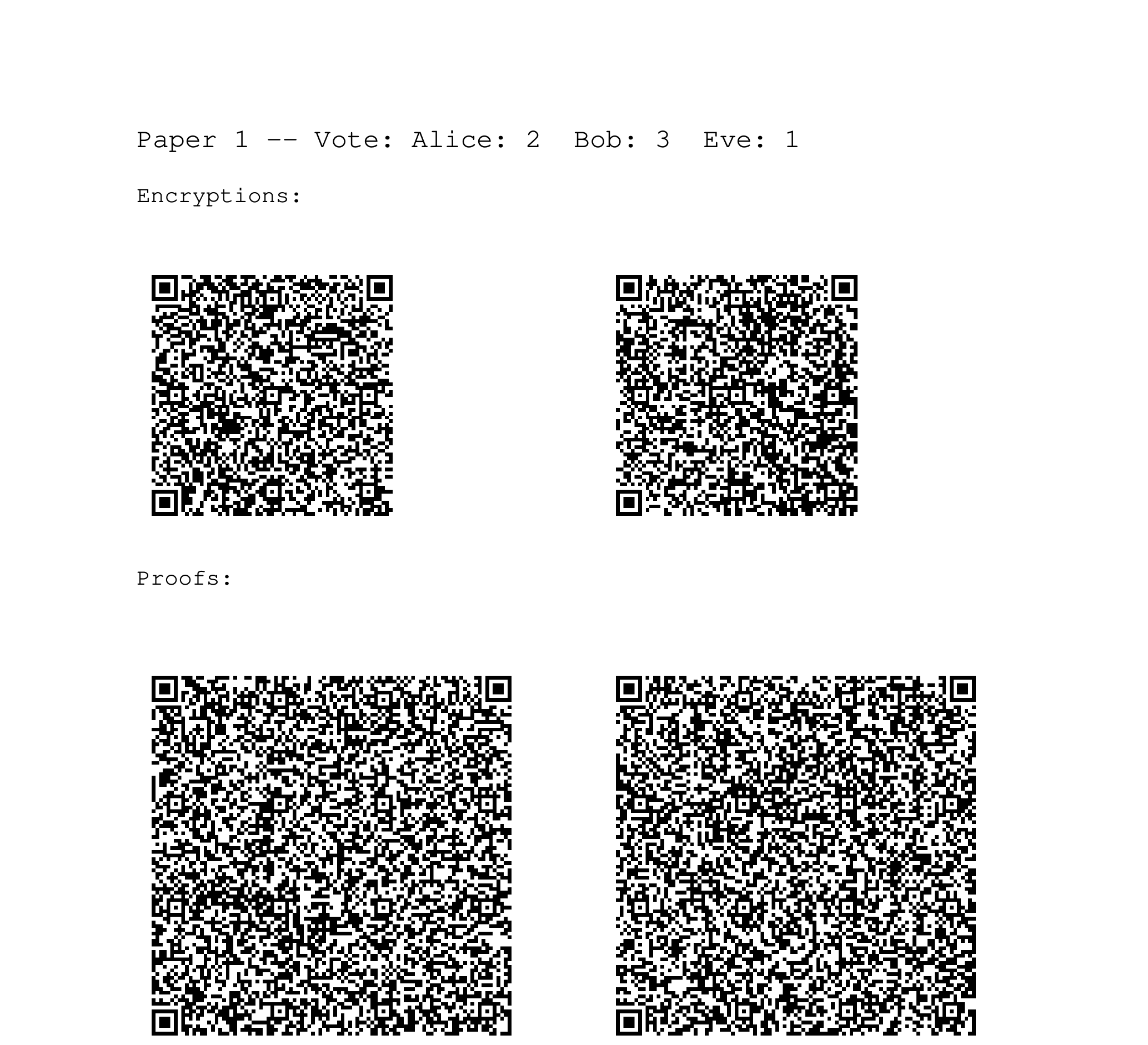}
	\caption{Paper 1: {\bf The voter only needs to check the plaintext vote at the top.}  This example is a ranking: Eve first, Alice next, Bob last. Each encryption QR code contains two ciphertexts, overall including $\{a, b, r_a, r_b\}_{pk}$. The proof QR code has the corresponding proofs of plaintext knowledge.}
	\label{fig:paper1}
\end{figure}

\begin{figure}
	\includegraphics[scale=0.4]{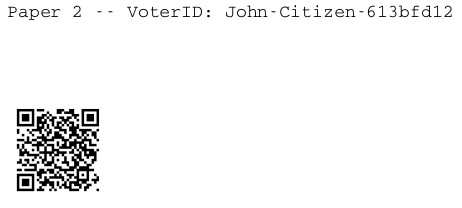}
	\caption{Paper 2: The voter just checks the plaintext $\vid$. 
	}
		\label{fig:paper2}
\end{figure}

\VTNote{Two nitpicks: 1. we should define [1,q-1].  2: I think it should be [1,q].}
\VTNote{Consider whether we should accept multiple registrations.  I think no?
If not, fix below.}
\begin{algorithm} 	\caption{\textit{Cast:} Vote generation and casting protocol 
	}
	\begin{algorithmic}[1]
	\State $Device$: $a,b,r_a,r_b\leftarrow \{1,\ldots,q-1\}$ (uniformly at random)\label{step:cast:secrets}
	\State $Device$: $c_a\leftarrow\commit_{\mathbb{P}}(a;r_a), c_b\leftarrow\commit_{\mathbb{P}}(b;r_b)$\label{step:cast:commitments}
	\State $Device\rightarrow \Wbb$: $\mathcal{B}^{\text{registered}}_\vid := (c_a, c_b)$ \label{Step:VoterCommit}
	\State $Voter\rightarrow Device$: $\Vote$\label{step:cast:choosevote}
	\State $Device$: $\Mac\leftarrow a\cdot \Vote+b\bmod q$
	\State $Device$:  Checks WBB for $c_a, c_b$. Sets $e_\Mac \gets \{g^\Mac\}_{pk}$, $e_\Vote, \gets \{g^\Vote\}_{pk}$, $e_\Param \gets \{a, b, r_a, r_b\}_{pk}.$
	\State $Device\rightarrow EC$: \label{step:voteSend}
	\State $\vid, e_\Mac, e_\Vote, \PrfKnow_{\mathbb{G}, pk}(e_\Mac), \PrfKnow_{\mathbb{G}, pk}(e_\Vote)$
	\State $EC\rightarrow \Wbb$:  $\mathcal{B}^{commit}_\vid := \big(\rerand(e_\Mac), \rerand(e_\Vote)\big)$\label{Step:ECPostsVoteMAC}
	\State $Device$: Checks WBB for $\rerand(e_\Mac), \rerand(e_\Vote)$\label{step:cast:wait}
	\State $Device\rightarrow \Paper_1$: $\Vote, e_\Param, \PrfKnow_{\mathbb{G}, pk}(e_\Param)$
	\State $Device\rightarrow \Paper_2$: $\vid$\label{step:cast:papers}
	\State $Voter\rightarrow EC$: $\Paper_1, \Paper_2$\ (by paper mail; $\Paper_1$ inside inner envelope)
	\end{algorithmic}
\label{fig:genAndCasting}
\end{algorithm}

\subsection{Receiving votes}
When the EC receives votes, they must be handled carefully to maintain voter privacy, so that the voter's identity is checked and forgotten before their vote is revealed. This seems complicated but is not much different from the common double-envelope system for protecting postal vote privacy, except for the matching of encrypted and plaintext VoterIDs.



The protocol (\textit{Process vote}) is shown in Algorithm~\ref{fig:tallyingEachBallot}. Scrutineers may be present to observe the plaintext-VoterID and plain-paper vote acceptance. The (outer) envelope is opened, and \textbf{only} $\Paper_2$ is removed. The received $\vid$ is checked against the identification on the (outer) envelope (e.g. the voter's name, address or signature); if this step fails, the entire ballot is placed in a reject pile (which scrutineers may see) and $\receivedvid$ is added to the WBB list $\mathcal{B}^{\text{rejected}}$. 

\begin{algorithm} 	\caption{\textit{Process vote:} Vote receiving protocol 
	}
	\begin{algorithmic}[1]
		\LineComment{Run for each ballot $(\Paper_1,\Paper_2)$ received by mail}
		\LineComment{Primed variables $\vid'$ are used to indicate the EC may receive different values to those the voter sent}
		\State $\Paper_2 \to EC:\ \vid'$
		\State EC: $\text{Checks }\vid'\text{ matches electoral roll}$
		\State EC: Retrieves $D_{\vid'}$ using $\vid'$
		\State EC: Joins $e_{\vid}$  (inner envelope of $D_{\vid'}$) to $\Paper_1$
		\State EC: $\text{Destroy }\Paper_2$,
		\State EC: Shuffle batches of $\Paper_1$ attached to $ e_{\vid'}$. 
		\State $\Paper_{1} \rightarrow EC:\ \ \Vote', e_{\vid'},e_\Param',\PrfKnow_{\mathbb{G}, pk}(e_\Param')$
		\State EC: $\text{Verifies }\PrfKnow_{\mathbb{G}, pk}(e_\Param')$.
		\label{Step:VerifyPfKnow}
		On failure, post $\vid'$ to $\mathcal{B}^{\text{rejected}}$ and skip ballot.
		\State $\textit{EC} \rightarrow \Wbb :$ 
		\State Add $\big(\Vote', e_{\vid}, \rerand(e_\Param')\big)$ to $\mathcal{B}^{\text{received}}$    \label{Step:ECPostsPaperInfo}
	\end{algorithmic}
	\label{fig:tallyingEachBallot}
\end{algorithm}

Next, the EC retrieves, or generates, the corresponding prepared envelope $D_\vid$.
The inner envelope containing the corresponding ciphertext $e_\vid$ is attached to $\Paper_1$  (e.g. by stapling, or by inserting into the unmarked inner envelope) \textbf{without looking at $\Paper_1$}. Finally, $\Paper_2$ is destroyed, leaving no plaintext link from $\vid$ to vote.


The resulting $(\Paper_1, e_\vid)$ pairs are shuffled physically to remove any link to the order in which envelopes were opened. Next they are opened. For each $\Paper_1$, the proofs of knowledge are verified, again putting the ballot in a reject pile and adding $e_\vid$ to $\mathcal{B}^{\text{rejected}}$ if it fails.
(These are shuffled and decrypted after opening all ballots.)
Finally, $(\Paper_1, e_\vid)$ pairs with verified proofs have their contents re-randomised and posted to the WBB list $\mathcal{B}^{\text{received}}$.

\subsection{Tallying} \label{subsec:tallying}


Tallying is shown in Algorithm~\ref{fig:tallyingVotes}, performed jointly by all trustees.  First we decrypt the secrets and check they are a correct opening of the commitments. We  then construct a second MAC from the committed vote on the WBB, and check that it matches the committed MAC.  If so, then with high probability the vote was cast as the voter intended, assuming that her secrets $a, b$ were not exposed. If any verifications fail, the trustees should mark the $\vid$ and go to the next vote.


\begin{algorithm} \caption{\textit{Tally votes:} Vote tallying protocol for trustees $T$}
	\begin{algorithmic}[1]
				\LineComment{Mix $\mathcal{B}^\text{received}$ to produce re-randomised encryptions of each item, permuted consistently, indicated by $\overline{e}_\Vote$ etc.}
		\State $T\rightarrow \Wbb: \ \mathcal{B}^{\text{received'}}=\Mix_{\mathbb{G}, pk}(\mathcal{B}^{\text{received}})$\label{Step:MixCast}
		\LineComment{Decrypt}
		\For {$\mathcal{B}^{\text{received'}}_i=(\overline{e}_\Vote, \overline{e}_\vid, \overline{e}_\Param)$}
		\State $T: \big((a, b, r_a, r_b), \DecProof_1\big) \gets \Dec_{\mathbb{G}, sk}(\overline{e}_\Param)$\label{Step:DecryptParams}
		\State $T: (\receivedvid, \DecProof_2) \gets \Dec_{\mathbb{G}, sk}(\overline{e}_\vid)$\label{Step:DecryptID}
		\State $T\rightarrow \Wbb$:
		\State $\mathcal{B}^{\text{mixed}}_i=\big(\overline{e}_\Vote, (a,b, r_a, r_b), \receivedvid, \DecProof_1, \DecProof_2\big)$\label{Step:RecVoterIDToWBB}
		\EndFor
		\LineComment{Join by matching $\vid$ to $\receivedvid$}  \label{Step:FirstWBBMsgs}
		\For {$\mathcal{B}^{\text{mixed}}_i=\big(\overline{e}_\Vote, (a,b, r_a, r_b), \receivedvid,\ldots\big)$ s.t. $\receivedvid$ is unique and doesn't appear in $\mathcal{B}^\text{rejected}$}
		\If {$\mathcal{B}^{registered}_\receivedvid$ is empty}
		\State skip to the next iteration.
		\EndIf
		\State $\Wbb\rightarrow T: (c_a, c_b) = \mathcal{B}^{\text{registered}}_\receivedvid$
		\If {$c_a=\commit_{\mathbb{P}}(a;r_a)$ and $c_b=\commit_{\mathbb{P}}(a;r_b)$}
		\State {$\mathcal{B}^{mixed}_i$  is a correct opening for $\mathcal{B}^{registered}_\receivedvid$}\label{Step:CheckOpen}
		\EndIf
		\EndFor 
		\LineComment{For each $\vid$ with one correct opening, recreate the MAC to check it matches the EC's committed one.}
		\For{ all \vid} 
		\If{$\mathcal{B}^{registered}_\vid$ or $\mathcal{B}^{commit}_\vid$ is empty}
	  	\State skip to the next iteration
	  	\EndIf
	  		\If {$\mathcal{B}^{registered}_\vid$ has a unique correct opening $\mathcal{B}^{mixed}_i$}
		\label{Step:uniqueCommitOpen}
		\State $\Wbb\rightarrow T:\ (e_\Vote, e_\Mac) = \mathcal{B}^{\text{commit}}_\vid, $\par
		\State $T\rightarrow \Wbb: \PEP_{\mathbb{G}, pk}\big(\overline{e}_\Vote, e_\Vote)$\label{Step:ECVotePEP}
		\State $\overline{e}_\Mac \gets (e_\Vote)^a + \{g^b\}_{pk}$
		\State $T\rightarrow \Wbb: \PEP_{\mathbb{G}, pk}(\overline{e}_\Mac, e_\Mac)$\label{Step:ECMACPEP}
		\If{plaintext equivalence proofs pass}
		\State $T\rightarrow \Wbb:$ Add $\left(\vid, e_\Vote\right)$ to $\mathcal{B}^{\text{accepted}}$\label{Step:acceptedVotes}
		\EndIf
		\EndIf
		\EndFor
		\LineComment{Mix and decrypt to produce final tally. $\hat{e}_\Vote$ represents the mix re-randomisation; we ignore $\hat{e}_\vid$.}
		
		\State $T\rightarrow \Wbb: \mathcal{B}^{accepted'}=\Mix_{\mathbb{G}, pk}(\mathcal{B}^{\text{accepted}})$\label{Step:MixAccepted}
		\For {$\mathcal{B}^{\text{accepted'}}_i=\hat{e}_\Vote$ on $\Wbb$}
		\State $T\rightarrow \Wbb: \mathcal{B}^{\text{tally}}_i=\Dec_{\mathbb{G}, pk}(\hat{e}_\Vote)$\label{Step:DecryptVote}
		\EndFor
	\end{algorithmic}
	\label{fig:tallyingVotes}
\end{algorithm}

\subsection{Verification protocols}
We define the vote to be \emph{cast} when the voter puts the envelope in the mail (or in a box at the EC), \emph{recorded} when the EC posts it (encrypted) on the bulletin board, and \emph{counted} when the list of decrypted votes is published on the WBB.\footnote{We omit questions of proper counting of complex ballots---when the list of accepted votes is public, we assume someone counts them properly.}

Each voter must check the printed paper vote, then use a device to verify that her ID appears in the final mix. Scrutineers verify the plain-paper aspects of the election. The WBB transcript is publicly verifiable.
The procedure for the WBB transcript verification is  in Algorithm~\ref{fig:globalVerification} (\textit{GlobalVerify}).

\subsubsection{By election scrutineers}
Scrutineers may observe the process of receiving paper ballots. This is not relevant to the security properties proven in this paper,\VTNote{Right? I can't see anywhere this is used in any of the security proofs.} but is relevant to the claim that the system falls back to traditional postal vote assumptions if the cryptography is broken, or if clients and EC are both dishonest.

Scrutineers present when the envelopes are opened must:\VTNote{I'm not sure that 2 and 3 add much.  We might be better off simply saying that they do whatever ID checking is traditional in their system, e.g. eligibility of the person on the outer envelope, or valid $\vid$ if that's equivalent.}
\begin{enumerate}
	\item verify the received $\vid$ is on the electoral roll and has not already had a vote included;
	\item verify the vote posted to the WBB (in step 9) matches the vote on $\Paper_1$.
\end{enumerate}


\subsubsection{By the voter}
The system is easy for a voter to verify. Before sending her vote by mail, the voter checks that the printed ballot paper matches the vote she intends to cast, and that the plaintext $\vid$ on Paper~1 is correct. 

\shorten{
The voter should either carry out the WBB transcript verification in Algorithm~\ref{fig:globalVerification}, or check that an organisation she trusts has done this verification.}

Once the receiving process is complete, she should check that her voter ID appears on the accepted list. 
The precise voter-verification protocol is shown in Algorithm~\ref{fig:voterVerification} (\textit{VoterVerify}). The voter must check the paper vote herself while casting it, but the WBB check can be outsourced to anyone.  Indeed, the voter only needs to check so that she can detect the non-arrival of her paper (or interference by an adversary).

\begin{algorithm} 	\caption{\textit{VoterVerify:} Voter's recorded-as-intended verification protocol}
	\begin{algorithmic}[1]
		\LineComment{The voter checks their printouts to verify}
		\State that the vote on $\Paper_1$ matches their intended vote, and		
		\State that the $\vid$ on $\Paper_2$ is correct.
		\LineComment{At the end of the election, the voter checks that their $\vid$ appears in $\mathcal{B}^{\text{accepted}}$ on the WBB.}
	\end{algorithmic}
	\label{fig:voterVerification}
\end{algorithm}

\subsubsection{Public Tally verification}
This consists of verifying the proofs that the tally protocol has been properly conducted. It is described in Algorithm~\ref{fig:globalVerification}.

\begin{algorithm}	\caption{\textit{GlobalVerify:}} Verification protocol for facts asserted on the WBB
	\begin{algorithmic}[1]
		\State Verify the mix proof for $\mathcal{B}^{\text{received}'}$ in Step~\ref{Step:MixCast} of \textit{Tally} \label{Step:VerFirstMix}
		\State Verify the decryption proofs in Steps \ref{Step:DecryptParams} and \ref{Step:DecryptID} of \textit{Tally} \label{Step:VerFirstDecrypt}
		\State Verify all PET proofs in Steps \ref{Step:ECVotePEP} and \ref{Step:ECMACPEP} of \textit{Tally} \label{Step:VerPETs}
		\State Verify the mix proof for $\mathcal{B}^{\text{accepted}'}$ in Step~\ref{Step:MixAccepted} of \textit{Tally} \label{Step:VerFinalMix}
		\State Verify the decryption proofs in Step~\ref{Step:DecryptVote} of \textit{Tally}	\label{Step:VerFinalDecrypt}
		\For {each row of $\mathcal{B}^{\text{registered}}$} \State Verify that $\vid$ is unique in $\mathcal{B}^{\text{registered}}$
		\EndFor
		\For {each row of $\mathcal{B}^{\text{commit}}$} \State Verify that $\vid$ is unique in $\mathcal{B}^{\text{commit}}$
		\EndFor
		\For {each row of $\mathcal{B}^{\text{accepted}}$}\label{Step:VerFacts}
			\State Verify that $\vid$ is unique in $\mathcal{B}^{\text{mixed}}$ and does not appear in $\mathcal{B}^{\text{rejected}}$ \label{Step:VerifyRecVoterIDToWBB}
			\State Verify that exactly one opening in $\mathcal{B}^{\text{registered}}$ is a correct opening for $c_a, c_b$
			\State Verify that the PETs in Steps~\ref{Step:ECVotePEP} and \ref{Step:ECMACPEP} of \textit{Tally} pass
		\EndFor
	\end{algorithmic}
	\label{fig:globalVerification}
\end{algorithm}

\subsection{Interpretation of the outcome}\label{Sec:Interpretation}
We need to be precise about the election outcome, since it depends on a combination of paper and electronic votes.  At the end of the election, the WBB transcript contains four sets:\footnote{For each of these lists one voter ID may have many corresponding WBB posts, but the list is a set so we don't count the same voter ID multiple times.}
\begin{enumerate}
	\item registered voter IDs $\mathcal{L}^{\text{registered}}$ drawn from $\mathcal{B}^{\text{registered}}$ , {\it i.e.} those who have uploaded a VoterID and commitments,
	\item received voter IDs $\mathcal{L}^{\text{received}}$ drawn from the plaintext ballots $\mathcal{B}^{\text{received}}$ posted by the EC in Step~\ref{Step:ECPostsPaperInfo} of \textit{Process Vote}, and decrypted in Step~\ref{Step:RecVoterIDToWBB} of \textit{Tally votes}.
	\item rejected voter IDs $\mathcal{L}^{\text{rejected}}$ drawn from the ballots $\mathcal{B}^{\text{rejected}}$ that arrived with invalid proofs, and
	\item 
	accepted voter IDs $\mathcal{L}^{\text{tally}}$ of those ballots $\mathcal{B}^{\text{tally}}$ that uniquely matched a registered voter's commitments, posted to the WBB in Step~\ref{Step:acceptedVotes} of \textit{Tally votes}.
\end{enumerate}

If $\mathcal{L}^{\text{tally}}$ is not a subset of $\mathcal{L}^{\text{registered}} \cup \mathcal{L}^{\text{received}}$, then something has gone badly 
wrong (and verification should fail). But in the normal course of an election we expect some deviation: some voters will register but never vote, some votes will go astray in the mail, or some votes will be misrecorded on arrival.  We want to devise a reasonable definition of an acceptable election outcome that can detect fraud but not cause the election to fail if small deviations are observed.  

The paper record consists of all ballots that passed traditional paper acceptance. It includes the votes in $\mathcal{B}^{\text{rejected}}$. 

The votes corresponding to the IDs in $\mathcal{L}^{\text{tally}}$ are those for which everything worked out perfectly---they should be accepted.  Their deviation from the plaintext paper ballots is an indication of one type of problem: possible substitution of paper ballots in the mail or by the EC.  Another type of problem is voters who registered but did not have a unique commitment match at Step~\ref{Step:uniqueCommitOpen} of \textit{Tally}, or did not pass PEPs in Step~\ref{Step:ECMACPEP} or Step~\ref{Step:ECVotePEP} of \textit{Tally}.  Depending on the exact nature of the problem, this could be evidence of attempted fraud \ifcompressing{}{(e.g. multiple commitment openings could indicate malware on the voter's device)} or a legitimate decision to register but not vote.  In summary, $\mathcal{L}^{\text{tally}}$ provides an arguable election outcome, while $\mathcal{L}^{\text{registered}}$, $\mathcal{L}^{\text{received}}$, and $\mathcal{L}^{\text{rejected}}$ provide some indication as to the extent of errors or manipulation attempts. Call the amount of detected error $\varepsilon=|\mathcal{L}^{\text{registered}} \cup \mathcal{L}^{\text{received}}|-|\mathcal{L}^{\text{tally}}|$.

Each democracy would have to decide how to deal with inconsistent results or evidence of problems.  Let $\mathcal{O}$ be the outcome of the election \emph{according to the paper record} (e.g. a tally of votes made for each candidate) with margin $M$  (e.g. half the difference in vote counts between the top two candidates). For a given WBB transcript $\tau$, define the acceptable number of caught errors to be $d$. One obvious formula would be: accept $\mathcal{O}$ if the demonstrated error in received votes was below the margin, (i.e. $d = M$). Another could be: accept $\mathcal{O}$ if the demonstrated error in received votes was below the margin, ignoring voters who registered but for whom a vote was not received (i.e. $d=M + |\mathcal{L}^{\text{registered}}| - |\mathcal{L}^{\text{received}}|$).

We abstract these choices out by defining the result to be:

\begin{equation*}
\Result(\tau, \mathcal{O}) =
\begin{cases}
\mathcal{O}, & \text{if $\varepsilon < d$ and }\textit{GlobalVerify}(\tau)\text{ passes} \\
\bot & \text{otherwise,} \\
&\text{ where } \varepsilon=|\mathcal{L}^{\text{registered}} \cup \mathcal{L}^{\text{received}}|-|\mathcal{L}^{\text{tally}}| \\
& \text{ and $d$ is determined by policy.} 

\end{cases}
\end{equation*}

To be confident that there were at most $d$ errors\footnote{In~\cite{kiayias2015end}, $d$ represents \textit{undetected} errors since the verification procedure is probabilistic. In our protocol we should be able to detect every error, so the interpretation is slightly different.}, at least $\theta=|\mathcal{V}|-(M-d)$ voters must correctly verify their votes (where $\mathcal{V}$ is the set of voters). Thus there is an inverse relationship between the allowed deviation and the number of voters that are allowed to not perform verification.
We assume 
\begin{itemize}
	\item the voter's receipt (later referred to as $\alpha_l$ for voter $\mathcal{V}_l$) consists of their $\vid$,
	\item there is some (out of scope) way for voters to check that their $\vid$ is unique, for example, it could be their name and address (which is bad for privacy, but ensures that it doesn't clash with someone else's),
	\item there is some (out of scope) way for observers to check that all registered voters are eligible.
\end{itemize}

We do not attempt to defend against denial-of-service attacks: votes can be scratched from the tally, {\it e.g.} if an adversary knows the target's $\vid$. However, this will be detected.
Preventions of this are a topic for future work.

\subsection{Sketch of security arguments} 
We provide below an English outline of the arguments we will use to prove security and privacy properties.

A client colluding with the EC or post can cheat, because then the adversary knows $a,b$, so it can change the paper ballot and generate a fake $(\Mac, \Vote)$ pair. 
In this threat model our system is no better than plain-paper postal voting.

Clearly a cheating EC can write a bad MAC, Vote or $\vid$ onto the WBB, rather than re-randomisations of what it received.  This will force the MAC match to fail and the electronic vote to be excluded.  A cheating client can do similarly. Each of these cheating individually will be detected.

We prove that if either the EC or the client is honest, then the vote cannot be substituted undetectably. Informally, suppose the client is honest, then consider what value the corrupt EC posts as that voter's re-randomised MAC and vote in Step~\ref{Step:ECPostsVoteMAC} of Cast.  If it is a valid MAC for the encrypted vote, then the EC must know $a,b$ or has been very lucky.  (The proof that it knows $a,b$ is that, if it could decrypt the encrypted Vote and MAC it received, it would know two different points on the line defined by $m = a\cdot v +b$.)  If it is not a valid MAC for that Vote, but somehow passes the PEP in the final step, then the cheating EC must have broken either the mix or the PEP.

This is formalised and proved in Section~\ref{subsec:EC-Verif}. Section~\ref{subsec:C-Verif} shows why the paper vote defends against a cheating client.

Informally, honest-but-remembering Receipt Freeness is achieved because for any vote she wishes to pretend to have cast, a voter can always generate a MAC consistent with her commitments to her $a$ and $b$ (which she can open honestly to a coercer).  Server-side re-randomisation achieves Receipt Freeness (as in~\cite{aditya2004efficient}), because the voter's client does not know the randomness used to generate the ciphertexts posted on the WBB.  However, this applies to a coercer who sees only the bulletin board, and does not defend against a coercer who corrupts the EC or taps the channel between the EC and the voting client. This is made more precise in Section~\ref{subsec:RF} after a privacy proof (with no client collusion) in Section~\ref{subsec:Privacy}.

\section{Security proofs} \label{sec:securityProofs}
\subsection{Privacy}  \label{subsec:Privacy}
We follow Kiayias et al. \cite{kiayias2015end} in defining voter privacy of an election via a \textit{Voter Privacy} game denoted by $G^{\mathcal{A},}_{priv}$ that is played between an adversary $\mathcal{A}$ and a challenger $\mathcal{C}$; we do not use a simulator in this version because the adversary is simply trying to guess which voter selection made it into the results. For clarity, we removed some qualifiers from the definitions which do not apply to the scheme we wish to prove. This simplification serves only to strengthen the definitions and improve readability. The game is parameterised by the security parameter $\lambda$, the number of voters $n$, and the number of candidates $m$. We consider a set of candidates $\mathcal{P}$, a set of voters $\mathcal{V}$, and a set of allowed candidate selections $\mathcal{U}$, and introduce an \textit{election evaluation function} $f(\langle \mathcal{U}_1,\ldots,\mathcal{U}_n \rangle)$ that outputs a vector whose $i$th index is the number of times candidate $\mathcal{P}_i$ was voted for ($\mathcal{P}_i$ may be a complex set of choices).
\begin{definition}[Privacy Game]
Denoted by $G^{\mathcal{A}}_{priv}(1^\lambda,n,m)$.
\begin{enumerate}
    \item $\mathcal{A}$ on input $1^\lambda,n,m$ chooses a list of candidates $\mathcal{P} = \{P_1,...,P_m\}$, a set of voters $\V = \{V_1,...,V_n\}$, and the set of allowed candidate selections $\mathcal{U}$, providing $\mathcal{C}$ with the sets $\mathcal{P},\mathcal{V},$ and $\mathcal{U}$.
    \item $\mathcal{C}$ flips a coin $b \in \{0,1\}$ and performs the \textit{Setup} protocol on input $(1^\lambda,\mathcal{P},\mathcal{V},\mathcal{U})$ to obtain $sk,(G, g,q,pk)$, providing $\mathcal{A}$ with $(G, g,q,pk)$.
    \item 
        The adversary $\mathcal{A}$ and the challenger $\mathcal{C}$ engage in an interaction where $\mathcal{A}$ schedules \textit{Cast} protocols of all voters which may run concurrently.  For each voter $V_l \in \V$, the adversary chooses whether $V_l$ is corrupted:
        \begin{itemize}
            \item If $V_l$ is corrupted, $\mathcal{A}$ plays the role of $V_l$ and $\mathcal{C}$ plays the role of $EC$ in the \textit{Cast} protocol.
            \item If $V_l$ is not corrupted, $\mathcal{A}$ provides two candidate selections $(\mathcal{U}^0_l,\mathcal{U}^1_l)$ to the challenger $\mathcal{C}$. They must do so such that $f(\langle\mathcal{U}_l^0\rangle_{V_l\in \HV})=f(\langle\mathcal{U}^1_l)_{V_l\in \HV}\rangle`$ where $\HV$ is the set of honest voters (that is, the election result w.r.t. the honest voters does not leak $b$).
			
			$\mathcal{C}$ operates on $V_l$'s behalf, using $\mathcal{U}^b_l$ as voter $V_l$'s input and playing the role of the $EC$. The adversary is allowed to observe the encrypted data $\{g^\Vote\}_{pk},\{g^\Mac\}_{pk}$ sent to the $EC$ in \textit{Cast}, as well as $\Wbb$, $\Paper_2$ and $\Paper_1$ (after shuffling). When the \textit{Cast} protocol terminates,  $\mathcal{C}$ provides to $\mathcal{A}$ the receipt consisting of the $\vid$ for voter $V_l$ (and therefore data on the WBB indexed by the $\vid$).
        \end{itemize}
    \item $\mathcal{C}$ performs the \textit{Tally} protocol playing the role of the election trustees.  $\mathcal{A}$ is allowed to observe the $\Wbb$.
    \item Finally, $\mathcal{A}$ using all the information collected above (including the contents of the $\Wbb$) outputs a bit $b^*$.
\end{enumerate}
Denote the set of corrupted voters as $\CV$ and the set of honest voters as $\HV = \V \setminus \CV$.  The game returns a bit which is 1 if and only if $b = b^*$.
We say that a voting scheme achieves voter privacy if for any PPT adversary $\mathcal{A}$:
\begin{align*}
    \left\vert Pr[G^{\mathcal{A}}_{priv}(1^\lambda,n,m) =1]-1/2\right\vert  = \negl(\lambda).
\end{align*}
\end{definition}
\paragraph{Assumptions}
For privacy, assume the following are honest:
\begin{itemize}
	\item the Electoral Commission (EC)
	\item the postal channel
	\item all threshold sets of the Election Tellers (ET)
	\item the voter's device
\end{itemize}
\begin{theorem}\label{thm:Privacy}
For any constant $m \in \mathbb{N}$ and $n=poly(\lambda)$, the e-voting system described in section 2 is private with respect to the privacy game $G^{\mathcal{A}}_{priv}(1^\lambda,n,m)$.
\end{theorem}

A proof is in Appendix~\ref{app:privacyProof}. It considers
the information visible to the adversary.
During \textit{Cast} the adversary sees
\ifcompressing{
$	\vid, c_a, c_b $ and $
\vid, \rerand(\{g^\Mac\}_{pk}, \{g^\Vote\}_{pk}).$}
{
\begin{subequations}
\begin{align*} 
	\vid, c_a, c_b \\
	\vid, \rerand(\{g^\Mac\}_{pk}, \{g^\Vote\}_{pk})
 \end{align*}
 \end{subequations}
}

During \textit{Tally} the adversary sees 
\begin{subequations}
\begin{align*}
 \Vote, \{\vid\}_{pk}, \{a,b,r_a,r_b\}_{pk}, \PrfKnow_{\mathbb{G}, pk}(\{a,b,r_a,r_b\}_{pk}) \\
(\ReceivedVote, \rerand\{\receivedvid\}_{pk}, \rerand\{a,b, r_a, r_b\}_{pk}) \\
(\{\ReceivedVote\}_{pk}, (a,b, r_a, r_b), \receivedvid, \text{decryption proof})\\
\Dec_{\mathbb{G}, pk}(\{g^\Vote\}_{pk})
\end{align*}
\end{subequations}  \VTNote{Strictly speaking, the above is not quite right because I haven't separated out the different ciphertexts, but that wouldn't fit so I'm leaving it as it is.}
Crucially, the adversary cannot use $\{\vid\}_{pk}$ in the above to relate $\Vote$ with $\vid$, since this relationship is forgotten when attaching $\{\vid\}_{pk}$  to $Paper_1$. Plaintext commitment openings are likewise not linkable to the plaintext vote.

\subsection{
\ifcompressing{Honest-but-remembering RF}
{Honest-but-remembering Receipt-freeness}} \label{subsec:RF}
We prove honest-but-remembering receipt-freeness (which is a stronger notion than privacy) against a weaker adversary.  

Consider a coercer who does not collude with the EC, but does make demands of the voting client.  We can prove only a passive form of receipt freeness, in which the colluding client follows the protocol honestly except for recording all its secrets.  We also have to assume that the channel to the EC is not tapped by the adversary, which models an attacker who does not have the capacity to intercept communications (such as TLS) over the Internet.\footnote{TLS is \emph{not} an untappable channel---people can prove what they sent by exposing the AES key.}  (At least one untappable channel in one direction is necessary and sufficient \cite{hirt2000efficient}, though we have two, and have not here considered an untrustworthy EC.)


\paragraph{Assumptions}
The receipt-freeness game $G^{\mathcal{A},\mathcal{S}}_{RecFree}(1^\lambda,n,m)$ is in Appendix~\ref{app:RF}. It is very similar to the privacy game, except that the adversary does not collude with the EC and cannot tap the channel between the voter and the EC.  
The adversary may view only the WBB.  It may, however, demand to see a (possibly faked) view from any honest voter. 

\paragraph{Setup}
The coercer will demand that the voter cast some vote $v$, and then provide the coercer with a transcript describing the setup, ballot generation and ballot casting for $v$.  

\paragraph{Proof main idea}
The Voter's coercion-resistance strategy is\ifcompressing{}{ simply} to truthfully reveal $a, b, r_a, r_b$ but claim to have sent
$\Mac_{\text{cr}} = a.v+b \bmod q$ as their MAC.  We rely on the re-randomised encrypted MAC that the EC posts on the WBB being indistinguishable from a re-randomised encryption of $\Mac_{\text{cr}}$.

\begin{theorem}
	For any constant $m \in \mathbb{N}$ and $n=poly(\lambda)$, the e-voting system described in section 2 has receipt freeness with respect to the game $G^{\mathcal{A},\mathcal{S}}_{RecFree}(1^\lambda,n,m)$
\end{theorem}

A proof is given in Appendix~\ref{app:RF}.

\paragraph{What this means in practice} The assumption that a coercer cannot tap the electronic channel from client to EC excludes adversaries associated with any network-based attacker, including those who see only encrypted TLS traffic. The assumption that the coercer cannot tap communications through the paper channel excludes an attacker who is physically present to watch the voter generate and post their vote.
It also assumes that a voter filming themselves creating and posting the envelope would not be convincing. We do not know how hard it is to fake such a video in practice, but note that our protocol does not add anything to such a video (such as specific ciphertexts) that would make it any more convincing than any other, except through collusion with the EC.

Honest-but-remembering receipt-freeness is better than no receipt freeness in the following practical scenario: suppose that the coercer has compromised the voter's computer, and seemingly has read access to all of their communications but doesn't know whether this access is genuine or simulated (e.g. if the voter is running the client in a virtual machine and controlling what the attacker sees).  With Helios, the attacker would be able to distinguish these two cases by verifying that the encrypted vote constructed by the device was posted to the bulletin board.  With our system, assuming the attacker can't verify what was sent to the EC, the system does not provide a way for the attacker to distinguish whether the voter sent the vote it seems to have sent, or intercepted it outside the coercer's view and sent something else. Hence read-only access does not allow coercion. However, if the coercer can instruct the voter to deviate from the protocol then coercion does succeed.

\subsection{Verifiability against a cheating EC} \label{subsec:EC-Verif}
Here we formalise the argument that, if the corrupt EC successfully posts a valid MAC for the claimed \textit{ReceivedVote} then it knows $a,b$, so this happens only for a negligible number of votes without client collusion. We use a modified version of the end-to-end verifiability game from~\cite{kiayias2015end}; our version does not allow the adversary to control the client and the EC simultaneously. The definition uses a \textit{vote extractor} algorithm $\mathcal{E}$, which given an election transcript $\tau$ and a set of honest voter receipts $\alpha_l$ outputs the set of \textit{dishonest} votes $\{\mathcal{U}_l\}_{\mathcal{V}_l\in\mathcal{V}\setminus\mathcal{\tilde{V}}}$. (We will use the metric $d_1$, meaning the absolute difference in number of votes for each candidate.)

\begin{definition}[EC Verifiability Game (after \cite{kiayias2015end})]
	We denote the game by $G^{\mathcal{A},\mathcal{E},d,\theta}_{EC-Ver}(1^\lambda, m,n)$.
	
	\begin{enumerate}
		\item $\mathcal{A}$ on input $1^\lambda,n,m,$ chooses a list of candidates $\mathcal{P} = \{P_1,...,P_m\}$, a set of voters $\V = \{V_1,...,V_n\}$, and the set of allowed candidate selections $\mathcal{U}$.  It provides $\mathcal{C}$ the sets $\mathcal{P},\mathcal{V},$ and $\mathcal{U}$.
		\item $\mathcal{A}$ performs the \textit{Setup} protocol on input $(1^\lambda,\mathcal{P},\mathcal{V},\mathcal{U})$ to obtain $sk,(G, g,q,pk)$, providing $\mathcal{C}$ with $(G, g,q,pk)$.
		\item
		The adversary $\mathcal{A}$ and the challenger $\mathcal{C}$ engage in an interaction where $\mathcal{A}$ schedules \textit{Cast} protocols of all voters which may run concurrently.  For each voter $V_l \in \V$, $\mathcal{A}$ can either completely control the voter or allow $\mathcal{C}$ to operate on their behalf, in which case $\mathcal{A}$ provides a candidate selection $\mathcal{U}_l$ to $\mathcal{C}$.  Then, $\mathcal{C}$ engages with the adversary $\mathcal{A}$ in the Cast protocol so that $\mathcal{A}$ plays the role of the EC and the postal service.  If the protocol terminates successfully, $\mathcal{C}$ obtains the receipt $\alpha_l=\vid$ on behalf of $V_l$.
		\VTNote{We need to be careful about what we're calling a receipt here.  We don't really have encrypted receipts in the usual E2E way here, but we don't want it to look as if your $\vid$ is generated when you vote, or we are subject to clash attacks.  Your $\vid$ needs to be demonstrably not the same as someone else's.  It may be there's no receipt as such, just an opportunity to observe your row of the WBB.}
		Let $\tilde{\mathcal{V}}$ be the set of honest voters (i.e. those controlled by $\mathcal{C}$) that terminated successfully.
		
		\item $\mathcal{A}$ posts the election transcript $\tau$ to the WBB.
	\end{enumerate}

	The game returns a bit which is 1 iff the following conditions are true:

\begin{enumerate}
	\item $\Big|\left\{l\in [n]\ |\ \textit{VoterVerify}(\alpha_l)\text{ passes}\right\}\Big| \geq \theta$ (i.e. at least $\theta$ honest voters verified successfully);
	\item $\Result(\tau, \mathcal{O}) \neq \bot$; and
	\item for the metric $d_1$ and election outcome function $f$:
				$$d_1(\Result(\tau, \mathcal{O}), f(\langle\mathcal{U}_1,\ldots,\mathcal{U}_n\rangle)) > d$$
		where $\{\mathcal{U}_l\}_{V_l \in \mathcal{V} \setminus \tilde{\mathcal{V}}} \leftarrow \mathcal{E}(\tau, \{\alpha_l \}_{V_l \in \tilde{\mathcal{V}}})$
		(That is, the deviation from the true result is larger than the accepted error $d$.)
\end{enumerate}
	We say that a voting scheme achieves EC verifiability if for any PPT adversary $\mathcal{A}$:
	\begin{align*}
	Pr\left[G^{\mathcal{A},\mathcal{E},d,\theta}_{EC-Ver}(1^\lambda,n,m) =1\right]= \negl(\lambda).
	\end{align*}
\end{definition}

\subsubsection{A simplified \ifcompressing{}{version of the} protocol for proving verifiability}
Consider a simplified version of the protocol in which there is only one decryption authority (this is not the privacy game after all). Remember that the attacker can modify the plaintext ballots as well as the EC's computations, though it does not control the voting client of the honest voters.

\begin{theorem}
For any constant $m\in\mathbb{N}$ and $n=poly(\lambda)$, a specified result function $\Result(\tau, \mathcal{O})$ defining a threshold $0 \leq d < M$ for an election with margin $M$, and $\theta=|\mathcal{V}|-(M-d)$, the simplified single-decryptor version of the protocol satisfies EC verifiability.
\end{theorem}
\begin{proof}
    We begin by defining the vote extractor $\mathcal{E}$. For each corrupt voter ID, it considers the commitment pair posted by the voter's device in Step~\ref{Step:VoterCommit} of \textit{Cast}, and the encrypted vote-MAC pair posted by the EC in Step~\ref{Step:ECPostsVoteMAC} of \textit{Cast}. It inspects the WBB transcript $\tau$ and outputs:
    \begin{enumerate}
        \item zero, if the $\vid$ has no matches in Step~\ref{Step:FirstWBBMsgs} of \textit{Tally votes} or no correct opening in Step~\ref{Step:uniqueCommitOpen}.
        \item zero, if the $\vid$ has more than one such match or correct opening\EMNote{Attack where the client prints bogus openings and sends someone else the real ones}
        \item zero, if there is a unique match and correct opening but either of the PETs in Steps~\ref{Step:ECVotePEP} and~\ref{Step:ECMACPEP} are not successful
        \item $\ReceivedVote$ otherwise.
    \end{enumerate}

    The first three cases correspond to a vote that was not submitted, or a verification failure. Case 4 represents successful verification of a vote that makes it into the tally. We will argue that the adversary has a negligible probability of successfully (and undetectably) substituting a vote with a different one in this case, and thus producing a deviation larger than the accepted error $d$. If the adversary can forge any of the zero-knowledge proofs, they have the ability to do this substitution; for example, a forged mix proof could allow many votes to be tampered with; a forged decryption proof could make a false claim about an encrypted vote. The soundness properties for these proofs guarantee the adversary has a negligible probability $\eta_1=\negl(\lambda)$ of doing so successfully. (This is where we use the adaptive soundness property, because the EC as prover chooses the ciphertext.)

	From here we assume the ZKPs are true; that is, the statement they assert is true, and there is some witness for each. 
	
	We walk backwards through the protocol. Each tallied vote in Step~\ref{Step:DecryptVote} of \textit{Tally votes} corresponds to:
	\begin{enumerate}
		\item a $\vid$ (via the mix and decryption proofs verified at Steps~\ref{Step:VerFinalMix} and~\ref{Step:VerFinalDecrypt} of \textit{GlobalVerify})
		\item secret parameters $a, b$ (via the ID matching verified at Step~\ref{Step:VerFacts} of \textit{GlobalVerify})
		\item a received vote (via the mix and decryption proofs verified at Steps~\ref{Step:VerFirstMix} and~\ref{Step:VerFirstDecrypt} of \textit{GlobalVerify})
		\item an encrypted MAC and vote from Step~\ref{Step:ECPostsVoteMAC} of \textit{Cast} posted \textbf{before the adversary knew} $a$ or $b$ (via the PETs verified at Step~\ref{Step:VerPETs} of \textit{GlobalVerify}, as well as the above mix and decryption proofs)
	\end{enumerate}

	Our attention turns to the commitments posted by the voter's device. (Remember that the BB is not under the adversary's control and hence the cheating EC cannot prevent the client from uploading its initial commitments.) Also note, that since \textit{GlobalVerify} passes only one commitment is present for each Voter and the voter's device checked that the commitment was the one it uploaded. At Step~\ref{Step:ECPostsPaperInfo} of \textit{Process Vote} (Algorithm~\ref{fig:tallyingEachBallot}), the EC must choose a particular vote and encrypted commitment openings $a, b$ to post alongside the vote and encrypted $\receivedvid$. There are three possibilities for such a commitment opening, compared to the commitment posted alongside $\vid$ in Step~\ref{Step:VoterCommit} of \textit{Cast}.

	\begin{enumerate}
		\item The opening may match the commitment.
		\item The opening may match a different voter's commitment.
		\item The opening may match no voter's commitment.
	\end{enumerate}

	Case 1 is the successful case where the correct commitment is opened; the security properties of Pedersen commitments guarantee the opening is legitimate except with only negligible probability $\eta_2=\negl(\lambda)$. Note that the EC cannot submit many possible openings and hope that one is a successful forgery --- the uniqueness condition in Step~\ref{Step:VerFacts} of \textit{GlobalVerify} prevents multiple attempted openings from being accepted. Case 2 will not pass verification, since only openings where $\receivedvid = \vid$ should be accepted in Step~\ref{Step:VerFacts} of \textit{GlobalVerify}. Similarly, Case 3 will not pass verification at the same step. We therefore discount the possibility of forged commitments for the remainder of the discussion.

	We now arrive at the key argument of the voting scheme. We will demonstrate that even a computationally-unbounded adversary cannot cheat in these circumstances with non-negligible probability. This adversary receives the genuine voter ID and ciphertexts $\{g^\Vote\}_{pk}, \{g^\Mac\}_{pk}$ during \textit{Cast}, which they can brute-force to produce plaintexts $\Vote, \Mac$. They will post encryptions of \textbf{different} values $\Vote_\text{cheat}, \Mac_\text{cheat}$ to the WBB in Step~\ref{Step:ECPostsVoteMAC} of \textit{Cast}. The PETs verified in Steps~$\ref{Step:VerPETs}$ and $\ref{Step:VerFacts}$ of \textit{GlobalVerify} (which we assume are honest) ensure that
	$$a\cdot\Vote_\text{cheat}+b=\Mac_\text{cheat}\text{ with }\Vote_\text{cheat}\neq\Vote$$
	
	But the adversary also knows that $a\cdot\Vote+b=\Mac$, and thus knows two points on the line defined by $a$ and $b$. The adversary has therefore extracted $a$ and $b$ from the information it had received by Step~\ref{Step:ECPostsVoteMAC} of \textit{Cast}, which included only one point on the line and two perfectly-hiding commitments to $a$ and $b$. However, given a fixed pair $a,b\in\{1,\ldots,q-1\}$, a vote, and a MAC there are $q-2$ other pairs
	$$a'=a+k,\ b'=b-k\cdot Vote\text{ for }k\in\{1,\ldots,q-1\}$$

	such that $a'\cdot\Vote+b'=\Mac$.

	Perfectly-hiding commitments leak no information; the adversary must therefore have guessed $a$ and $b$. Since $a$ and $b$ were chosen uniformly at random, the adversary can do so with probability $\frac{1}{q-1}$.

	We are left with three ways the adversary can succeed: by forging ZKPs (with probability $\eta_1$), by forging commitments (with probability $\eta_2$), or by forging MAC/vote pairs (with probability $\frac{1}{q-1}$). If the adversary does not forge a ZKP, it must forge commitments or MAC/vote pairs for at least $d$ votes --- but the probability for forging these for even one vote is negligible. All told, any PPT adversary must therefore have advantage at most
	$\eta_1+\eta_2+\frac{1}{q-1}=\negl(\lambda).$
\end{proof}

\subsection{Recorded-as-cast Verifiability against a cheating client} \label{subsec:C-Verif}
We claimed that the protocol allowed voters to detect cheating against an adversary who controls either the client or (the postal service and the EC) but not both.  We therefore assume that if the client is corrupted, the paper ballot is properly received and processed at the electoral commission. We assume that $\theta=|\mathcal{V}|-(M-d)$ voters check their plain paper printout with their vote and their VoterID.  Note that the \emph{Voter} is honest, but an honest voter's client may be malicious.

The client verifiability game is defined in Appendix~\ref{app:clientVerif} and is very similar to that in \cite{kiayias2015end}. The cheating clients win if
the election tally is accepted, but is substantially different from the true result.  The proof relies on the assumption that the paper ballot is properly posted and processed.

\begin{theorem}
	For any constant $m\in\mathbb{N}$ and $n=poly(\lambda)$, a specified result function $\Result(\tau, \mathcal{O})$ defining a threshold $0 \leq d < M$ for an election with margin $M$, and $\theta=|\mathcal{V}|-(M-d)$, the simplified single-decryptor version of the protocol satisfies client verifiability.
	
\end{theorem}

The definition of the game and proof of the Theorem are given in Appendix~\ref{app:clientVerif}.

\section{Implementation and pilot}\label{sec:implementation}
We implemented a prototype in Rust.
The ElGamal cryptosystem was implemented over the prime-order Ristretto subgroup of Curve25519 using \texttt{curve25519-dalek}~\cite{curve25519dalek}.

For the shuffle, we re-implemented a variant of the Verificatum shuffle~\cite{verificatum2018} where each row contains multiple ciphertexts, based on the presentation in~\cite{haenni2017pseudo}. The implementation was tested for efficiency, as a real-world system may need to handle millions of votes. Using an Intel i7-10750H ;mobile CPU to run a shuffle on 100000 rows, each with 6 ciphertexts, the code was able to generate a shuffle and corresponding proof in 38.34 seconds, and was able to verify this proof in 26.43 seconds. Practicality is a key benefit of our protocol---for $n$ votes, we require only $O(n)$ PEPs, and the shuffle proof requires only $O(n)$ elliptic curve additions and multiplications.

A real-world pilot of the protocol was run with three trustees
A small number of volunteers were asked to rank candidates Alice, Bob, and Eve---the example votes from Section~\ref{subsec:voting} were taken from this pilot.  Seven ranked-choice votes were submitted and physically mailed to one of the authors acting in the role of EC. Five of them were scanned and tallied. The other two were (unintentionally) lost due to human error in the testing process, which gave us a nice opportunity to test our basic verification steps. The verification protocol successfully verified four of the votes and revealed that one was missing. The other two voters did not use the verification protocol to check their votes.
Although this is nothing close to a full usability study or realistic test, it shows that our system is complete and that simple failures can be detected.

\section{Limitations and possible enhancements} \label{sec:extensions} \label{sec:limitations}
\paragraph{Trust in the Electoral Commission for privacy}
Our protocol trusts the EC to forget the link between $\vid$ and the encrypted $\vid$ it generates.
In traditional postal voting there is analogous trust in the Electoral Commission to not misbehave, for example, by opening both the inner and outer envelope at the same time. However, this can be observed by scrutineers without infringing on the secrecy of the ballot, and is therefore a relatively low, observable, risk. \ifcompressing{}{As such, the trust assumption is restricted to an Electoral Commission covertly opening ballots, inspecting them, and then resealing them prior to counting, which cannot be discounted entirely, but is an attack that is both more complicated and more risky.}
Trusting an electronic component not to leak is much more problematic.

\paragraph{Non-collusion between EC and client}
Another important limitation is the assumption that the client and the EC are not both compromised.  In an ideal world, people would download and compile an open-source voting app from an independent entity they trusted.  In practice voters generally get their voting instructions and software from the same EC that will be receiving their votes.  (Two of the authors worked on an end-to-end verifiable e-voting project in which the electoral authority refused to issue any cast-as-intended verification instructions at all.)  This is an important practical question for the true security of our scheme.  However, most verifiable e-voting systems suffer from some version of the same problem, hence the huge motivation to fall back to plain paper mail.

\paragraph{Revealing which voters' MAC matched}
The current version of this protocol reveals  $\vid$s during the MAC matching process.  This produces public information about who cast a valid vote and who didn't.  The protocol could be altered to hide this information, thus making it secret which MACs matched, though the total numbers would be public.  This would change the verifiability property from an individual to a group one: voters would not be able to tell whether their own vote had been dropped, though everyone would be able to see the total numbers of dropped and invalid votes.
\VTNote{This needs a fair bit more thought, because there's a big difference between seeing the invalid vote you intended to cast fail matching, and seeing someone else's vote fail matching when you didn't know whether it was intended to be valid.  The point is, however, that it only matters for dropping, not for substituting, votes: substitution is possible if the client leaks a,b and impossible (with non-negl prob) otherwise, regardless of how we organise the matching.}

\paragraph{Trusting the client for privacy}
\EMNote{TODO: Talk about other ways to decouple/couple paper 1 and 2.}
The client is trusted for privacy: although a client controlled by the voter can lie to a coercer, a client controlled by the coercer knows which MAC was submitted and hence which vote was sent.

However, the MAC generation, vote encryption, and data upload (to the EC) steps of \textit{Cast}
do not all need to be performed by the same device.  One device could generate ciphertexts (without knowing which were uploaded), and a different device could upload them (without knowing their contents).  This has significant advantages for privacy against the client, effectively splitting the information about how the person voted between two devices.  \ifcompressing{}{A malicious client would not know which vote had been submitted.}

Also, since the voter's ability to defend against a cheating EC relies on the client to keep the random $a$ and $b$ values secret, it would be beneficial to expand the protocol so that $a$ and $b$ were generated in a distributed way by multiple devices.  \ifcompressing{}{However, this significantly complicates the user experience.}

\VTNote{Maybe just leave Xavier's improved/edited version for a later paper.}

\paragraph{Accountability}
The current protocol emphasises verifiability but does not attempt to provide accountability or defence against voters who maliciously claim that there was a problem when there was none.  The most obvious attack is to post fake votes with someone else's VoterID.  This will not forge a vote, but it will appear as an indication of a problem.  Although this isn't a verifiability failure, it would  cause votes to be cancelled (detectably) and to give an impression of fraud. A standard authentication mechanism would address this.  For example, voters could be given a secret nonce without which an upload for their VoterID would not be accepted to the WBB.

\paragraph{WBB}
We also make strong assumptions on the existence and security properties of a web bulletin board.  In practice, this would be implemented in some specific way involving either threshold trust on a set of peers, or a final verification step requiring some extra work from each voter.

\section{Conclusion}
We provide a significant step forward in the design of verifiable remote voting protocols.  Our protocol combines very simple cast-as-intended verification of a plaintext printout with strong guarantees of verifiability for those observers who choose to check.  We also provide a (passive, honest-but-remembering) version of receipt-freeness against an attacker who cannot tap either of the voter's communication channels.  This is the first work to provide all these advantages.

For politically binding elections, trusting the Electoral Commission for privacy, and trusting it not to collude with the client for verifiability, may be unacceptable. However, for non-political elections, where independent scrutineers may not be present and possibly even double envelopes are not used, the trust assumptions may be equivalent to what they are in typical postal voting. In such scenarios, the additional benefit of verifiability remains valuable. \VTNote{pithier ending please}
\newpage
\bibliographystyle{plain}

\bibliography{bibliography.bib}

\appendix

\section{Proof of Theorem~\ref{thm:Privacy}}\label{app:privacyProof}
We prove Theorem~\ref{thm:Privacy} (vote privacy) assuming a threshold of trustees is honest.

\paragraph{Intuition} The intuition is that we are going to replace all of the ciphertexts, except those directly submitted by the adversary, with encryptions of nonsense without the adversary noticing.  Specifically in game 3, the adversary's view during \textit{Cast} with respect to the honest voters is going to be as follows plus some random ciphertexts and simulated proofs:
\begin{subequations}
	\begin{align} 
	(\vid, c_a, c_b) \\
	\vid \\
	\vid
	\end{align}
\end{subequations}
During \textit{Tally} the adversary will see:
\begin{subequations}
	\begin{align}
	\Vote \\
	\ReceivedVote \\
	((a,b, r_a, r_b), \receivedvid) \\
	g^\Vote
	\end{align}
\end{subequations}
The mixing that EC performs between (3c) and (4a), between (4b) and (4c), and also between (4c) and (4d) hides the order in which the votes were submitted, preventing the vote from being trivially matched with the voter ID.

\begin{proof}
Define the advantage between game $G_i$ and $G_j$ to be
$$\mathsf{Adv}_{G_i,G_j}(\mathcal{A}) := \frac{1}{2}\big\vert Pr[\mathcal{A}=1| G_i]-Pr[\mathcal{A}=1| G_j]\big\vert$$
Consider the following sequence of games.
\begin{description}
	\item[Game] $G_0$: The actual game $G^{\mathcal{A}}_{priv}(1^\lambda,n,m)$.	By definition $Adv_{G_0,G^{\mathcal{A}}_{priv}(1^\lambda,n,m)}(\mathcal{A}) = 0$.
	
	\item[Game] $G_1$: Let $G_1$ be the same as Game $G_0$ except that all PETs and decryptions are simulated using knowledge of the plaintext rather than decryption keys. This is allowable since all ciphertexts being decrypted or tested for plaintext equivalence are either produced by the challenger, or they are produced by the adversary but are accompanied by zero-knowledge proofs of knowledge (and therefore the challenger can extract them with the zero-knowledge extractor).  By the soundness properties of the zero-knowledge proof of knowledge proofs $Adv_{G_1,G_0}(\mathcal{A}) = \negl(\lambda)$.
	
	\item[Game] $G_2$: Let $G_2$ be the same as Game $G_1$ except all the ZKPs used to demonstrate correct mixing, correct decryption, and correct PETs that the challenger performs are simulated via the zero-knowledge simulator. At this point, the challenger no longer uses the secret key corresponding to $pk$ for any purpose. The mixing must be simulated to avoid leaking information as to the permutation or randomness used.\EMNote{Boneh \& Shoup Thm 19.10 proves this for equality of discrete logs.}\THNote{Ballot independence is effectively utilised in games 1 and games 2 without being explicitly mentioned.  More technically we utilise the fact that the adversary knows the votes contained in the ciphertexts she submitted.}
	Since the proofs of knowledge are non-malleable and the EC filters for duplicates, none of the adversary's proofs depend on the simulated proofs; we can thus continue to use the extractor on these proofs. By their zero-knowledge properties, $Adv_{G_2,G_1}(\mathcal{A}) = 0$.
	
	\item[Game] $G_3$: Let $G_3$ be the same as Game $G_2$ except that all the values to be encrypted are replaced by random values from an oracle and all values to be re-encrypted are replaced with fresh encryptions of random values. Decryption and plaintext equivalence is always simulated, so we never provide a decryption oracle; this allows us to rely on the IND-CPA property of ElGamal, guaranteeing that $Adv_{G_3,G_2}(\mathcal{A}) = \negl(\lambda)$.
\end{description}

In Game $G_3$, the ciphertexts and proofs contain random values with the exception of the (decoupled) $\vid$s and $\Vote$s.  The $\vid$s and $\Vote$ are decoupled as a result of the mixing which applies a random permutation to these lists---this is secret assuming the honesty of at least one mixing trustee and the honesty of $n - k + 1$ decrypting trustees (for $k$-out-of-$n$ secret sharing).  Recall that the second criterion of the game says that the adversary loses if the set of honest votes leaks $b$.  Therefore, the adversary cannot have any advantage in winning $G_3$. Following the chain of games yields
$$Adv_{G_3, G^{\mathcal{A}}_{priv}(1^\lambda,n,m)}=\negl(\lambda)$$

so $\mathcal{A}$'s advantage in $G^{\mathcal{A}}_{priv}(1^\lambda,n,m)$ is negligible.
\end{proof}

We note that it is possible to trust the voter's device less and the EC more by changing the adversary's view and the privacy proof follows in much the same manner.  This change only affects the definition in point 3; the adversary is allowed to observe $a$ and $b$ but not $\Paper_2$.

\section{Definition and proof of honest-but-remembering Receipt Freeness}\label{app:RF}

Informally, the adversary is attempting to coerce a voter $V_l$ into submitting a vote for candidate selection $\mathcal{U}^0_l$.
The game is defined as follows.

\begin{definition}[Honest-but-remembering Receipt-freeness Game]
	We denote the game by $G^{\mathcal{A},\mathcal{S}}_{RecFree}(1^\lambda,n,m)$.
	\begin{enumerate}
		\item $\mathcal{A}$ on input $1^\lambda,n,m,$ chooses a list of candidates $\mathcal{P} = \{P_1,...,P_m\}$, a set of voters $\V = \{V_1,...,V_n\}$, and the set of allowed candidate selections $\mathcal{U}$.  It provides $\mathcal{C}$ the sets $\mathcal{P},\mathcal{V},$ and $\mathcal{U}$.
		\item $\mathcal{C}$ flips a coin $b \in \{0,1\}$ and performs the \textit{Setup} protocol on input $(1^\lambda,\mathcal{P},\mathcal{V},\mathcal{U})$ to obtain $sk,(G,g,q,pk)$, providing $\mathcal{A}$ with $(G,g,q,pk)$.
		\item 
		The adversary $\mathcal{A}$ and the challenger $\mathcal{C}$ engage in an interaction where $\mathcal{A}$ schedules \textit{Cast} protocols of all voters which may run concurrently.  For each voter $V_l \in \V$, the adversary chooses whether $V_l$ is corrupted:
		\begin{itemize}
			\item If $V_l$ is corrupted, they engage in a \textit{Cast} protocol where $\mathcal{A}$ plays the role of $V_l$ and $\mathcal{C}$ plays the role of $EC$.
			\item If $V_l$ is not corrupted, $\mathcal{A}$ provides two candidate selections $(\mathcal{U}^0_l,\mathcal{U}^1_l)$ to the challenger $\mathcal{C}$.  They must do so such that $f(\langle\mathcal{U}_l^0\rangle_{V_l\in \HV})=f(\langle\mathcal{U}^1_l\rangle_{V_l\in \HV})$ where $\HV$ is the set of honest voters (that is, the election result w.r.t. the honest voters does not leak $b$).
			
			\hspace{0.25in}$\mathcal{C}$ operates on $V_l$'s behalf, using $\mathcal{U}^b_l$ as the voter $V_l$'s input.  \emph{The adversary is allowed to observe $\Wbb$ only,} where $\mathcal{C}$ plays the role of $V_l$ and the $EC$.  When the \textit{Cast} protocol terminates, the challenger $\mathcal{C}$ provides to $\mathcal{A}$:
			\begin{enumerate}
				\item the receipt consisting of the $\vid$ for voter $V_l$, and 
				\item if $b=0$, the current view, $\text{view}_l = (a,b,r_a,r_b,\Vote,\Mac,\{g^\Vote\}_{pk},\{g^\Mac\}_{pk},$\newline$
				\{a\}_{pk},\{b\}_{pk},\{r_a\}_{pk},\{r_b\}_{pk},\{\vid\}_{pk})$, of the voter $V_l$ that the challenger obtains from the $\textit{Cast}$ execution. If $b=1$, the challenger instead provides a simulated view of the internal state of $V_l$ produced by $\mathcal{S}(\text{view}_l)$.\footnote{Intuitively, if $b=0$ the voter honestly gives its view to the adversary. If instead $b=1$ the voter simulates a fake view and gives that to the adversary, voting however they please.}
				
			\end{enumerate}
		\end{itemize}
		\item $\mathcal{C}$ performs the \textit{Tally} protocol playing the role of the election trustees $ET$.  $\mathcal{A}$ is allowed to observe the $\Wbb$.
		\item Finally, $\mathcal{A}$ uses all the information collected above (including the contents of the $\Wbb$) to output a bit $b^*$.
	\end{enumerate}
	Denote the set of corrupted voters as $\CV$ and the set of honest voters as $\HV = \V \setminus \CV$.  The game returns a bit which is 1 if and only if $b = b^*$.
	
	We say that a voting scheme achieves receipt-freeness if there is a PPT voter simulator $\mathcal{S}$ such that for any PPT adversary $\mathcal{A}$:
	\begin{align*}
	\left\vert Pr[G^{\mathcal{A},\mathcal{S}}_{RecFree}(1^\lambda,n,m) =1]-1/2\right\vert  = \negl(\lambda).
	\end{align*}
\end{definition}

The proof of receipt-freeness relies on defining a coercion-resistance strategy in which the voter tells the truth about the secret values $(a,b)$ it has committed to, but lies about the vote and then claims a MAC corresponding to the claimed vote and the truthful $(a,b)$.  We show that for a receipt-freeness adversary this is indistinguishable from obedience.

Since we do not rely on secrecy of  the $a,b$ values for receipt freeness, we can use an intermediate game in which $a$ and $b$ are known to the EC.  

We now restate and prove the main theorem.

\begin{theorem}
	For any constant $m \in \mathbb{N}$ and $n=poly(\lambda)$, the e-voting system described in section 2 has honest-but-remembering receipt freeness with respect to the game $G^{\mathcal{A},\mathcal{S}}_{RecFree}(1^\lambda,n,m)$
\end{theorem}

\begin{proof}
	We briefly recap the information visible to the adversary.
	During \textit{Cast} the adversary sees
	\begin{subequations}
		\begin{align*} 
		(\vid, c_a, c_b) \\
		(\vid, \rerand\{g^\Mac\}_{pk}, \rerand\{g^\Vote\}_{pk})
		\end{align*}
	\end{subequations}
	
	The main difference from the privacy game is that the adversary may demand the voter's view, including secret information. After \textit{Cast}, $\mathcal{A}$ sees the possibly-simulated view
	\begin{multline*}
	(\vid,a,b,r_a,r_b,\Vote,\Mac,\{g^\Vote\}_{pk},\{g^\Mac\}_{pk},\\ 
	\{a\}_{pk},\{b\}_{pk},\{r_a\}_{pk},\{r_b\}_{pk},\{\vid\}_{pk})
	\end{multline*}
	
	During \textit{Tally} the adversary sees 
	\begin{subequations}
		\begin{align*}
		(\ReceivedVote, \rerand\{\receivedvid\}_{pk}, \hspace{3cm}\\
		           \rerand\{
		\{a\}_{pk},\{b\}_{pk},\{r_a\}_{pk},\{r_b\}_{pk} \}) \\
		(\{\ReceivedVote\}_{pk}, (a,b, r_a, r_b), \receivedvid, \text{deryption proof})\\
		\Dec_{\mathbb{G}, pk}(\{g^\Vote\}_{pk})
		\end{align*}
	\end{subequations}
	
	\paragraph{Defining the simulator}
	The simulator $\mathcal{S}$ for each honest voter $V_l$ receives the voter's view (including candidate selections $(\mathcal{U}_l^0,\mathcal{U}_l^1)$ and randomness for all the encryptions)
	\begin{multline*}
	(\vid,a,b,r_a,r_b,\Vote=\mathcal{U}_l^1,\Mac,\{g^\Vote\}_{pk},\\ \{g^\Mac\}_{pk}, 
	\{a\}_{pk},\{b\}_{pk},\{r_a\}_{pk},\{r_b\}_{pk},\{\vid\}_{pk})
	\end{multline*}
	Then $\mathcal{S}$  outputs the fake view
	\begin{multline*}
	(\vid,a,b,r_a,r_b,\Vote'=\mathcal{U}_l^{0},\Mac' = a\cdot\Vote' + b, \{g^{\Vote'}\}_{pk}\\ \{g^{\Mac'}\}_{pk},
	\{a\}_{pk},\{b\}_{pk},\{r_a\}_{pk},\{r_b\}_{pk},\{\vid\}_{pk})
	\end{multline*}
	
	Define the advantage between game $G_i$ and $G_j$ to be
	$$\mathsf{Adv}_{G_i,G_j}(\mathcal{A}) := \frac{1}{2}\big\vert Pr[\mathcal{A}=1| G_i]-Pr[\mathcal{A}=1| G_j]\big\vert$$
	
	Consider the following sequences of games. 
	\begin{description}
		\item[Game] $G_0$: The actual game $G^{\mathcal{A},\mathcal{S}}_{RecFree}(1^\lambda,n,m)$, where the challenger uses $\mathcal{U}^b_l$ in the \textit{Cast} protocol and the above simulator is invoked when $b=1$.  (That is, voters vote as they wish and run the coercion-resistance strategy.)
		
		By definition $Adv_{G_0,G^{\mathcal{A},\mathcal{S}}_{RecFree}(1^\lambda,n,m)}(\mathcal{A}) = 0$.
		
		\item[Game] $G_1$: The same as Game $G_0$, except the decryptions and plaintext equivalence tests are simulated with knowledge of the plaintext as in Theorem~\ref{thm:Privacy}; $Adv_{G_1, G_0}(\mathcal{A})=\negl(\lambda)$.
		
		\item[Game] $G_2$: The same as Game $G_1$, except the proofs used to demonstrate correct decryption, plaintext equivalence, and correct mixing are simulated with their zero-knowledge simulators as in Theorem~\ref{thm:Privacy}. The challenger replaces the re-randomised ciphertexts from the mixing with fresh encryptions to ensure the link is destroyed. We have $Adv_{G_2, G_1}(\mathcal{A})=\negl(\lambda)$.
		
		\item[Game] $G_3$: The same as Game $G_2$, except when $b = 1$:
		\begin{enumerate}
			\item In Step~\ref{Step:ECPostsVoteMAC} of \textit{Cast}, the challenger posts an encryption of the claimed MAC, $\{g^{\Mac'}\}_{pk}$, instead of a re-randomised encryption of the actual MAC $\{g^\Mac\}_{pk}$.
			\item In Step~\ref{Step:ECPostsPaperInfo} of \textit{Process Vote}, the challenger changes the posted (re-randomised) encryptions of $\receivedvid$ and $a, b, r_a, r_b$ so that they appear together with the votes they claimed to have cast.
		\end{enumerate}
		\textit{Tally} can then proceed as usual; we have changed the votes and MACs consistently so that they are still plaintext-equivalent. Since all we have done is change encryptions for which the adversary does not know the randomness and the mixing breaks the link between successive encrypted votes, the IND-CPA property of ElGamal yields $Adv_{G_3, G_2}(\mathcal{A})=\negl(\lambda)$.
		
		\item[Game] $G_4$: The same as Game $G_3$, except the challenger (acting as the honest voters) ignores the value of $b$ and always obeys the adversary. Since the adversary does not see anything different to what it saw in Game $G_3$, $Adv_{G_4, G_3}(\mathcal{A})=0$.
	\end{description}
	The adversary has no advantage in Game $G_4$ because the value of $b$ is ignored. Following the chain of games then yields
	$$Adv_{G_4, G^{\mathcal{A},\mathcal{S}}_{RecFree}(1^\lambda,n,m)}=\negl(\lambda)$$
	so $\mathcal{A}$'s advantage in $G^{\mathcal{A},\mathcal{S}}_{RecFree}(1^\lambda,n,m)$ is negligible.
\end{proof}

\section{Client Verifiability definition and proof} \label{app:clientVerif}

\begin{definition}[Client Verifiability Game (after \cite{kiayias2015end})]
	We denote the game by $G^{\mathcal{A},\mathcal{E},d,\theta}_{C-Ver}(1^\lambda, m,n)$ for a vote extractor algorithm $\mathcal{E}$ (which may be super-polynomial).
	
	\begin{enumerate}
		\item $\mathcal{A}$ on input $1^\lambda,n,m,$ chooses a list of candidates $\mathcal{P} = \{P_1,...,P_m\}$, a set of voters $\V = \{V_1,...,V_n\}$, and the set of allowed candidate selections $\mathcal{U}$.  It provides $\mathcal{C}$ the sets $\mathcal{P},\mathcal{V},$ and $\mathcal{U}$.
		
		\item $\mathcal{A}$ performs the \textit{Setup} protocol on input $(1^\lambda,\mathcal{P},\mathcal{V},\mathcal{U})$ to obtain $sk,(G,g,q,pk)$, providing $\mathcal{C}$ with $(G,g,q,pk)$.
		
		\item
		The adversary $\mathcal{A}$ and the challenger $\mathcal{C}$ engage in an interaction where $\mathcal{A}$ schedules \textit{Cast} protocols of all voters which may run concurrently.  For each voter $V_l \in \V$, $\mathcal{A}$ can either completely control the voter or allow $\mathcal{C}$ to operate on their behalf, in which case $\mathcal{A}$ provides a candidate selection $\mathcal{U}_l$ to $\mathcal{C}$.  Then, $\mathcal{C}$ engages with the adversary $\mathcal{A}$ in the Cast protocol so that $\mathcal{A}$ plays the role of the voting client.  The postal system and the EC execute honestly.  If the protocol terminates successfully, $\mathcal{C}$ obtains the receipt $\vid$ on behalf of $V_l$.
		\VTNote{We need to be careful about what we're calling a receipt here.  We don't really have encrypted receipts in the usual E2E way here, but we don't want it to look as if your $\vid$ is generated when you vote, or we are subject to clash attacks.  Your $\vid$ needs to be demonstrably not the same as someone else's.  It may be there's no receipt as such, just an opportunity to observe your row of the WBB.}
		
		Let $\tilde{\mathcal{V}}$ be the set of honest voters (i.e. those controlled by $\mathcal{C}$) that terminated successfully.
		
		\item Finally, the (honest) EC posts the election transcript $\tau$ to the WBB.
	\end{enumerate}
	
	The game returns a bit which is 1 iff the following conditions are true:
	
	\begin{enumerate}
		\item $\Big|\left\{l\in [n]\ |\ \textit{VoterVerify}(\alpha_l)\text{ passes}\right\}\Big| \geq \theta$ (i.e. at least $\theta$ honest voters verified successfully);
		\item $\Result(\tau, \mathcal{O}) \neq \bot$; and
		\item for the metric $d_1$ and election outcome function $f$:
		$$d_1(\Result(\tau, \mathcal{O}), f(\langle\mathcal{U}_1,\ldots,\mathcal{U}_n\rangle)) > d$$
		where $\{\mathcal{U}_l\}_{V_l \in \mathcal{V} \setminus \tilde{\mathcal{V}}} \leftarrow \mathcal{E}(\tau, \{\alpha_l \}_{V_l \in \tilde{\mathcal{V}}})$
		(That is, the deviation from the true result is larger than the accepted error $d$.)
	\end{enumerate}
	
	We say that a voting scheme achieves client verifiability if for any PPT adversary $\mathcal{A}$:
	\begin{align*}
	Pr\left[G^{\mathcal{A},\mathcal{E},d,\theta}_{Client-ver}(1^\lambda,n,m) =1\right]  = \negl(\lambda).
	\end{align*}
\end{definition}

\begin{theorem} \label{thm:ECVerifiability}
	For any constant $m\in\mathbb{N}$ and $n=poly(\lambda)$, a specified result function $\Result(\tau, \mathcal{O})$ defining a threshold $0 \leq d < M$ for an election with margin $M$, and $\theta=|\mathcal{V}|-(M-d)$, the simplified ZKP-based version of the protocol satisfies client verifiability.
	
\end{theorem}

\begin{proof}
	We assume that $\theta=|\mathcal{V}|-(M-d)$ voters run Algorithm~\ref{fig:voterVerification}
	(Voter verification) correctly, including checking whether their vote is included in $\mathcal{B}^{\textit{accepted}}$. These are the honest voters, though note that their client may be controlled by the adversary. The vote extractor $\mathcal{E}$ is the same as in the proof of Theorem~\ref{thm:ECVerifiability}.

	
    Consider one honest voter with a possibly-malicious client.
	Consider the following cases, of which the first three are the only ways of making a malformed ballot.
	\VTNote{Check that these first three are the only ways of making a malformed ballot.}
	Note, of course, that none of them prevent the cheating client from also sending other ballots either electronically or by colluding with another voter.
	
	\begin{description}
		\item[Case 1] Suppose the client printed on $\Paper_1$, or sent to the EC in Step~\ref{step:voteSend} of Algorithm~\ref{fig:genAndCasting}, at least one proof $\PrfKnow_{\mathbb{G}, pk}(e_\Param)$, $\PrfKnow_{\mathbb{G}, pk}(e_\Mac)$ or $\PrfKnow_{\mathbb{G}, pk}(e_\Vote)$, that contains a false statement but passes verification. By the adaptive soundness of the proofs, this is done with negligible probability $\eta_1$.
		
		\item[Case 2] Suppose client printed on $\Paper_1$, or sent to the EC in Step~\ref{step:voteSend} of Algorithm~\ref{fig:genAndCasting}, at least one proof $\PrfKnow_{\mathbb{G}, pk}(e_\Param)$, $\PrfKnow_{\mathbb{G}, pk}(e_\Mac)$ or $\PrfKnow_{\mathbb{G}, pk}(e_\Vote)$ that does not pass verification.
		
		These are checked in  Step~\ref{Step:VerifyPfKnow} of \textit{Process vote} (Algorithm~\ref{fig:tallyingEachBallot}).   Since we assume an honest EC, the VoterID will be included in $\mathcal{B}^{\text{rejected}}$, and the vote will not pass voter verification (Algorithm~\ref{fig:voterVerification}).

		\item[Case 3] Suppose the client printed on $\Paper_1$ values $\{a\}_{pk},\{b\}_{pk},\{r_a\}_{pk},\{r_b\}_{pk}$ that are not valid commitment openings of the $c_a$, $c_b$ commitments posted in Step~\ref{Step:VoterCommit} of Algorithm~\ref{fig:genAndCasting} (where $\Paper_1$ here means the one that the honest voter verified, though the dishonest client may have printed other values onto other instances of $\Paper_1$ and fraudulently inserted them into the post).
		In this case, either there will be no commitment opening in 
		Step~\ref{Step:CheckOpen}, or there will be multiple matching VoterIDs in Step~\ref{Step:uniqueCommitOpen}, 
		of \textit{Tally votes} / Algorithm~\ref{fig:tallyingVotes}.  The honest EC will therefore not add VoterID to $\mathcal{B}^{\text{accepted}}$, so the vote will not pass verification.
		
		\item[Case 4] Everything else.  
	\end{description}
	
	To count Case 4, observe that cases 1--3 are the only ways a ballot ($\textit{Paper}_1$) can be malformed. We can hence suppose for this case that the paper ballots are well-formed.  By assumption of postal/EC honesty, both the plaintext $\ReceivedVote$ and the encrypted $\receivedvid$ must be posted properly on the WBB at line~\ref{Step:ECPostsPaperInfo} of Algorithm~\ref{fig:tallyingEachBallot}, because that is what it specifies when the ZKPs are valid. So there is at least one valid commitment opening received and posted.

	We now count the number of $\receivedvid$'s in each possible category for  Step~\ref{Step:FirstWBBMsgs}  of \textit{Tally votes}.

If there are multiple valid  commitment openings for VoterID at Step~\ref{Step:uniqueCommitOpen}, $\vid$ will not appear in $\mathcal{B}^{\textit{accepted}}$.  If there are other matching RecVoterIDs, \textit{GlobalVerify} fails (Step~\ref{Step:VerifyRecVoterIDToWBB}).  

So now assume there is a unique valid commitment opening at Step~\ref{Step:uniqueCommitOpen}
and exactly one matching RecVoterID. By EC honesty, it must match what the voter checked on the paper ballot. If the cheating client did not send true encryptions of the vote and MAC at Step of \textit{Cast}, the plaintext equivalence tests checked at Steps~\ref{Step:ECMACPEP} and~\ref{Step:ECVotePEP} will fail and $\vid$ will not appear in $\mathcal{B}^{\textit{accepted}}$.

Therefore, if the voter runs the verification protocol correctly, either verification fails, or their vote has been  correctly included in the tally, except with negligible probability $\eta_1$.

	  
	We have computed the probability that one client cheated without detection, which implies that the probability that the client of any honest voter cheated without detection is also negligible.  
	Since all but $\theta=|\mathcal{V}|-(M - d)$ voters verified their vote successfully this is also an upper bound on the adversary's success probability.

\end{proof}

\end{document}